\documentclass[journal]{IEEEtran}

\ifCLASSINFOpdf
\else
\fi

\usepackage{amsthm}
\usepackage{makecell}
\usepackage{graphicx}
\usepackage{amssymb}
\usepackage{amsmath}
\usepackage{algorithm}
\usepackage{algorithmic}
\usepackage{subcaption}
\usepackage{booktabs} 
\usepackage{cite}
\usepackage{amssymb} 
\usepackage{nomencl}
\usepackage{extarrows}
\usepackage{hyperref}
\makenomenclature

\usepackage{etoolbox} 
\renewcommand\nomgroup [1] { %
  \item [ \bfseries 
  \ifstrequal{#1}{A}{\textbf{A. Abbreviations}}{ %
  \ifstrequal{#1}{P}{\textbf{B. Parameters}}{ %
  \ifstrequal{#1}{V}{\textbf{C. Variables}}{}}} %
] }

\theoremstyle{definition}
\newtheorem{theorem}{Theorem}
\newtheorem{proposition}[theorem]{Proposition}

\theoremstyle{definition}
\newtheorem{definition}{Definition}
\newtheorem{remark}{Remark}

\newtheorem{assumption}{\textit{Assumption}}

\allowdisplaybreaks

\begin{document}

\title{\huge Co-Scheduling of Energy and Production in Discrete Manufacturing Considering Decision-Dependent Uncertainties}

\author{Yiyuan~Pan,
        Zhaojian~Wang
\thanks{This work was supported by the National Natural Science Foundation of China (62103265), and the Young Elite Scientist Sponsorship Program by China Association for Science and Technology (No.YESS20220320). (\textit{Corresponding author: Zhaojian Wang})	}
\thanks{Y. Pan, and Z. Wang are with the Key Laboratory of System Control, and Information Processing, Ministry of Education of China, Department of Automation, Shanghai Jiao Tong University, Shanghai 200240, China, (email:wangzhaojian@sjtu.edu.cn). }
}


\maketitle

\begin{abstract}
Modern discrete manufacturing requires real-time energy and production co-scheduling to reduce business costs. In discrete manufacturing, production lines and equipment are complex and numerous, which introduces significant uncertainty during the production process. Among these uncertainties, decision-dependent uncertainties (DDUs) pose additional challenges in finding optimal production strategies, as the signature or the shape of the uncertainty set cannot be determined before solving the model. However, existing research does not account for decision-dependent uncertainties (DDUs) present in discrete manufacturing; moreover, current algorithms for solving models with DDUs suffer from high computational complexity, making them unsuitable for the real-time control requirements of modern industry. To this end, we proposed an energy-production co-scheduling model for discrete manufacturing, taking into account decision-dependent uncertainties (DDUs). Subsequently, we devised a method for transforming the constraints associated with DDUs into linear form ones. Finally, we designed a novel algorithm based on the column-and-constraint generation (C\&CG) algorithm and undertook a theoretical analysis of its performance of convergence and algorithmic complexity. A real-world engine assembly line was used to test our model and algorithm. Simulation results demonstrate that our method significantly reduces production costs and achieves better frequency regulation performance.
\end{abstract}


%
\IEEEpeerreviewmaketitle

\section*{Nomenclature}
\addcontentsline{toc}{section}{Nomenclature}
\subsection{Abbreviations}
\begin{IEEEdescription}[\IEEEusemathlabelsep\IEEEsetlabelwidth{$\mathcal{S}_3, \mathcal{R}_2$}]
	\item[FR] Frequency regulation.
	\item[RTP] Real-time price.
	\item[DER] Distributed energy resource.
	\item[PV] Photovoltaic.
	\item[BESS] Battery energy storage system.
\end{IEEEdescription}
\subsection{Variables}
\begin{IEEEdescription}[\IEEEusemathlabelsep\IEEEsetlabelwidth{$I^h_{n,p}, E^h$}]
	\item[$I^h_{n,p}$] Working state of cell $n$'s equipment $p$ at hour $h$.
	\item[$G^h_n$] The output of cell $n$ hour $h$.
	\item[$C^h_n$] The input of cell $n$ hour $h$.
	\item[$T^h$] Total time of equipment use at hour $h$.
	\item[$E^h$] Total power consumption at hour $h$.
	\item[$\tilde{T}^h$] Actual time of equipment use at hour $h$.
	\item[$\tilde{E}^h$] Actual power consumption at hour $h$.
	\item[$B^h_m$] Cargo amount in buffer $m$ at hour $h$.
	\item[$E_{EU}^h$] Compensatory power from BESS at hour $h$.
	\item[$E_{LU}^h$] Compensatory power from EPN at hour $h$.
	\item[$E_{SU}^h$] Power from EPN to BESS at hour $h$.
	\item[$S^h$] BESS electrical energy stock at hour $h$.
	\item[$B^h_{ss}$] Sales of by-products at hour $h$.
\end{IEEEdescription}
\subsection{Parameters}
\begin{IEEEdescription}[\IEEEusemathlabelsep\IEEEsetlabelwidth{$\underline{P}_{k}, \bar{P}_{k}$}]
	\item[$t_{n,p}$] Cell $n$'s Equipment $p$ single use time cost.
	\item[$T_{(n,p)}$] Cell $n$'s Equipment $p$ minimum time limit.
	\item[$e_{n,p}$] Cell $n$'s Equipment $p$ single use power cost.
	\item[$\alpha$] No-DDU model's yield rate.
	\item[$\alpha_{n,p}^h$] DDU model's yield rate for cell $n$'s equipment $p$ at hour $h$.
	\item[$k_n$] Cell $n$'s location.
	\item[$N_n$] Cell $n$ Maximum of uses per day.
	\item[$k_m$] Single transport volume of buffer $m$.
	\item[$\delta t$] Single transport time for buffer $m$.
	\item[$M_s$] Set of secondary product output nodes.
	\item[$\bar{S}$] BESS upper bound on electricity stock.
	\item[$\underline{P}_{k}, \bar{P}_{k}$] Climbing constraints lower, upper bound.
	\item[$\eta$] No-DDU model's FR parameter.
	\item[$\eta^h$] DDU model's FR parameter at hour $h$.
	\item[$\zeta$] No-DDU model's propensity to sell secondary products.
	\item[$\zeta^h$] DDU model's propensity to sell secondary products at hour $h$.
	\item[$s_m, s_s$] Selling price of main and secondary products.
	\item[$\lambda_{mis}$] Per-kwh regulation mismatch penalty.
	\item[$J$] Number of cell units in a battery.
	\item[$V_{nom}$] Nominal voltage of cell units in a battery.
	\item[$\beta_1 - \beta_7$] Coefficients of battery degradation.
	\item[$\delta_t'$] Battery degradation time interval.
	\item[$r_t, d_t$] Coefficient of FR demand.
	\item[$E_{ex}^h$] Expected power consumption at hour $h$.
	\item[$\mathcal{S}_1$] ambiguous set DDU.
	\item[$\mathcal{S}_2$] One factor probability distribution DDU.
	\item[$\mathcal{S}_3$] Multi-factors probability distribution DDU.
	\item[$\mathcal{R}_2$] Normal expected data distribution set.
	\item[$\mathcal{R}_2$] DDU-related expected data distribution set.
	\item[$f_b(\cdot)$] Battery degradation penalty function.
	\item[$f(\cdot)$] Inprecise Dirichlet model function.
	\item[$\gamma_1, \gamma_2$] DDU-related probability distribution correction parameters.
	\item[$LC_i^h$] Type $i$ combination of production line data at hour $h$ from historical data.
	\item[$w_i^h$] Ratio of secondary and main products for combination $i$ from historical data.
	\item[$\theta_i^h$] Probability of combination $i$ at hour $h$.
\end{IEEEdescription}

\section{Introduction}
%
%
%
%
\subsection{Backgroound} \label{Background}
The discrete manufacturing sector has experienced significant growth in recent years, driven by the globalization of supply chains and the rise of customization \cite{wang2019modelling}. Unlike process manufacturing, such as the chemical industry, discrete manufacturing—like engine assembly—typically involves a greater number of machines and more flexible and interchangeable process sequences. This results in a higher degree of freedom at the scheduling level, allowing for more complex decision-making and adjustments during production \cite{may2017energy}. This has led to a growing interest in the co-scheduling problem based on production and energy in discrete manufacturing. However, the inherent uncertainty of the discrete manufacturing industry presents a significant challenge to efficient scheduling \cite{parente2020production, dowling2018economic, kern2021peak}.

The uncertainties inherent to discrete manufacturing can be classified into two categories: decision-independent uncertainties (DIUs) and decision-dependent uncertainties (DDUs). Modern discrete manufacturing industrial parks typically consist of a production system and an energy system. The production system includes production lines and warehousing facilities, while the energy system comprises distributed energy resources (DER) generation equipment and energy storage systems. In the industrial park, DIUs originate from the intrinsic attributes of the equipment or the production environment, for instance, the quality of the raw materials or the daily power generation of solar panels. In contrast, DDUs are uncertainties related to production decisions. The following are the main types of DDUs in industrial parks: (1) The inclination towards producing main or by-products is related to uncertain order status and future production conditions, and it's also connected to current production decisions; (2) Product yield is influenced by both inherent equipment uncertainties and the choice of the production line; (3) Power frequency regulation requirements in production objectives are constrained by actual electricity usage decisions and are linked to unknown requirements form the power utility as well.

A production and energy co-scheduling model for discrete manufacturing industries can be modeled as a two-stage robust optimization model. In the production process, production profit serves as the primary objective, while energy requirements (such as frequency regulation and minimizing battery degradation) are considered secondary. As a result, this scheduling process is modeled as a two-stage optimization problem. Additionally, due to the presence of uncertainties, including DIUs and DDUs, the model becomes a robust decision-making problem.

\subsection{Related Works}
Existing research on uncertainties in the scheduling process of discrete manufacturing primarily focuses on decision-independent uncertainties (DIUs). Work \cite{DIUU1} concentrates on discrete manufacturing industries with a focus on DIUs of fluctuations in product and customer demand. It proposes a two-dimensional model to assess the ability to cope with these DIUs. Work \cite{DIUU2} focuses on the DIU of fluctuating daily production quantities and employs the non-dominated sorting adaptive differential evolution (NSJADE) algorithm to obtain a robust order scheduling solution. Work \cite{DIUU3} focuses on the uncertainty of machine failures, which are modeled using an exponential distribution and then handled through Markov chains. The uncertainty addressed in work \cite{DIUU4} is the examination of machine yield as a fuzzy variable, which is characterized by a probability-of-failure-working-time equation. Work \cite{DIU5} is concerned with DDU of performance degradation (e.g., equipment failures) within discrete manufacturing. To address these uncertainties, a simultaneous update method combining Adaboost, DNN, and LSTM is proposed. However, some uncertainties within these DIUs, such as equipment defect rates \cite{DIUU4, DIU5}—which are not only related to the inherent uncertainties of the equipment but also to the combination of upstream production equipment—should be modeled as decision-dependent uncertainties (DDUs).

Current algorithms for solving two-stage robust optimization models considering DDUs are primarily based on iterative decomposition algorithms or scenario tree-based partitioning algorithms. In work \cite{ddu1}, recursive optimization methods and numerical enumeration techniques are employed to solve the multi-stage models with the presence of DDUs. Work \cite{ddu2} investigates the two-stage robust optimization with polyhedral DDUs and proposes an iterative algorithm based on Benders decomposition. The algorithm employs optimality and feasibility cuts to address the coupling between uncertainty and decision-making. Work \cite{ddu3} initially transforms the multi-stage optimization problem containing DDU into its dualized form, subsequently employing the adversarial cutting plane algorithm to resolve the model. Work \cite{ddu4} employs a scenario tree based on model predictive control to transform the optimization problem containing DDUs into a series of subproblems to be solved under different learning scenarios. Work \cite{ddu6} uses an adaptive reliability improvement unit commitment (ARIUC) algorithm to efficiently decompose the DDU-related model into multiple sub-problems for solving. In conclusion, works \cite{ddu1, ddu2, ddu3} employ iterative decomposition algorithms, while works \cite{ddu4, ddu6} utilize partitioning approaches. However, both of these algorithms (iterative decomposition algorithms and scenario tree-based partitioning algorithms) exhibit high computational complexity, making them unsuitable for the real-time and rapid scheduling needs of modern industrial applications.

\subsection{Methodology}
To sum up, existing work on the co-scheduling of industrial production and energy, considering the presence of uncertainties, presents the following two issues about modeling and solving: 1) \textit{Modeling:} Existing work does not examine the discrete manufacturing scheduling problem that encompasses DDUs; rather, it focuses exclusively on DIUs. However, some of the DIUs therein are pertinent to decision-making processes and should be modeled as DDUs. 2) \textit{Solving:} In contrast to previous methodologies for DDU-related two-stage problems that focused solely on solvability, the modern industry requires algorithms that converge with low computational complexity, thereby enabling rapid and real-time control to achieve optimal production results.

In light of the aforementioned issues, we proposed a two-stage robust optimization model considering a complete set of three types of DDUs for discrete manufacturing, and a novel algorithm with performance guarantees is also designed. First of all, we constructed a model comprising multiple DDUs. Then, we simplify the constraints associated with the aforementioned three DDUs to linear forms using methods based on ambiguous sets, the imprecise Dirichlet model, and Cantelli's inequality, respectively. The problem was subsequently reduced to a mixed integer quadratic program (MIQP) problem. Subsequently, we proposed the decision-dependent C\&CG (DDCCG) algorithm based on the traditional C\&CG algorithm and provided a performance analysis for its convergence and optimality. Finally, we use the Petri net to represent a real-world engine assembly line and test our model and algorithm on this study case. The findings demonstrate that the industrial park reduces costs, enhances the anti-interference capabilities of the production line, and responds to frequency regulation and peak shaving requirements. These results substantiate the superiority of the proposed model and validate the reliability of the underlying algorithm.

In summary, the main contributions of this paper include:

\begin{itemize}
    \item This paper presents a two-stage robust optimization model for the discrete manufacturing industry that considers multiple DDUs. In particular, the model incorporates product yield DDU, frequency regulation penalty DDU, and product structure DDU into the energy and production co-scheduling process. In comparison to previous work in this field, this model more accurately reflects the interdependence between decisions and uncertainties in real-world industrial parks.
    \item We presented three paradigms for reducing different DDU-related constraints to a linear form. In particular, the three types of DDUs in discrete manufacturing are mathematically represented as ambiguous sets, univariate distributions, and multivariate distributions, which encompass the majority of real-world DDUs. Therefore, the methods presented in this paper can be extended and applied to other engineering scenarios in the real world.
    \item We proposed the DDCCG algorithm for the decision-dependent two-stage robust optimization model and proved its convergence and optimality theoretically. The original C\&CG algorithm has a fast convergence speed but cannot adapt to the decision-dependent model. Our work removes this limitation and can be adapted to the needs of modern industry.
\end{itemize}

The remainder of this paper is structured as follows. Section \ref{Section2} introduces the basic form of the optimization problem. Section \ref{Section3} presents modeling and linearization techniques for DDUs. Section \ref{Section4} describes the subsequent simplification of the problem and algorithm design. Section \ref{Section5} is case analysis. The last section summarizes our work.

\begin{figure}[!t]
\centering
\includegraphics[width=3in]{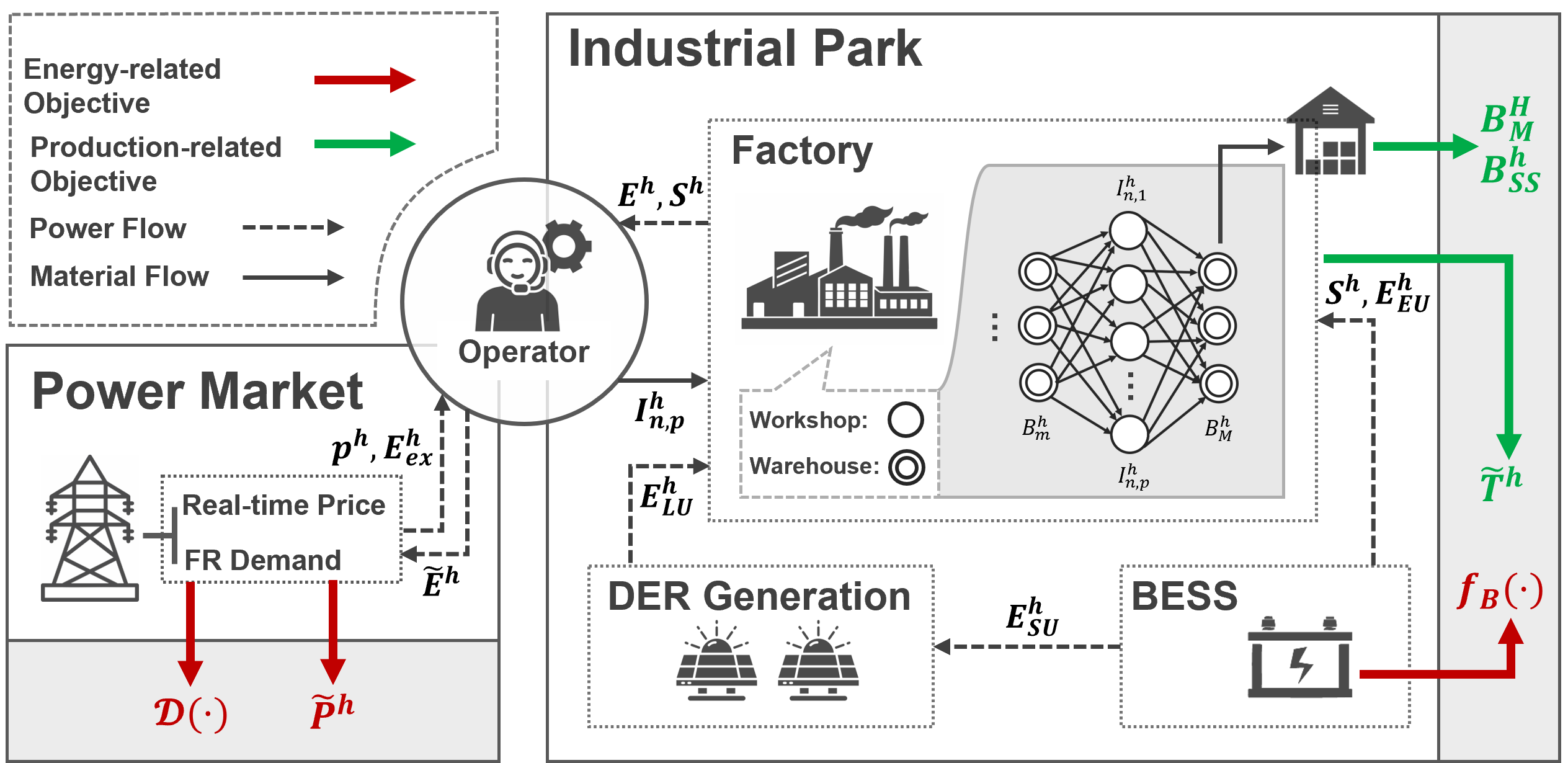}
\caption{Schematic diagram of the overall model architecture.}
\label{Fig1}
\end{figure}
\section{Problem Formulation} \label{Section2}
\subsection{Problem Overview}
A factory operator serves as the solver for a two-stage robust optimization problem, making decisions based on data collected from the power grid (e.g., electricity prices, frequency regulation requirements) and production information from the industrial park, aiming to maximize production profit.

A modern discrete manufacturing industrial park comprises a production system and an energy system. Without loss of generality, the production system includes production workshops and warehousing facilities \cite{cps}. It is modeled as a graph $G:=(V,E)$, where $V=\{1,...,N; N+1,...,N+M\}$ is the set of nodes. $N$ represents the number of production workshops, and $M$ represents the number of warehousing facilities. $E \subseteq V \times V$ is a set of directed edges. The production workshops provide various services, while the warehousing facilities handle the transfer of semi-finished products. For simplicity, the energy network includes a distributed energy resources (DER) generation device and a battery energy storage system (BESS), both of which are connected to the production network. The operator can access the energy network to supply power to the equipment. The scheduling time range is $T=\{1, ... ,H\}$.

\subsection{Objective Function}
As the entity responsible for the maintenance of the industrial park, our objective is to implement strategies that will minimize the total economic cost. The objective function can be expressed as follows
\begin{multline}\label{1}
    \min_{I^h_{n,p},E^h_{EU},E^h_{LU}}\sum\nolimits_{h=0}^H\tilde {T}^h\cdot r+f_B(S^h) \\
    + \sum\nolimits_{h=0}^H\tilde{P}^h+\sum\nolimits_{h=0}^H E^h_{fr} \\
    -B_M^H\cdot s_m-\zeta^h \sum\nolimits_{h=0}^HB^h_{ss}\cdot  s_s
\end{multline}

 In (\ref{1}) The first term is the equipment usage loss (intuitively, time of use and the maintenance costs are proportional.), the second term is the battery degradation penalty function \cite{bat1, bat2}, the third term is the power purchase cost, the fourth term is the frequency regulation penalty term \cite{fr}, and the fifth and sixth terms are the main products and by-products sales. The following equations (\ref{2a}) and (\ref{2b}) show the specific expression for the degradation function \cite{bat1, bat2}.
\begin{subequations}
\begin{gather}
    f_B(S^h)=J\lambda _{cell}V_{nom}h_{cell}(S^h)\Delta t',\forall t   \label{2a} \\
    h_{cell}(S^h)=\beta S^h,\forall t \label{2b}
\end{gather}
\end{subequations}

\begin{remark}
    (\textit{DDU variable $\zeta^h$}) The factory's production preference for primary and secondary products is a decision-dependent uncertainty (DDU) and this tendency is characterized by the weight factor $\zeta^h$. The factory's production tendencies for different products are related to actual production conditions. However, due to the inherent uncertainty in future order expectations and production conditions, the above tendency is DDU.
\end{remark}

\subsection{Production Model and Constraint}
Without loss of generality, the following matrix is used to describe the production-related variables in modern industries.
\begin{equation}
    \begin{bmatrix}
      \text{ID}_{w1}  & I_{w1,p}^h & G_{w1}^h & C_{w1}^h & t_{{w1},p} & e_{{w1},p} & \alpha_{w1,p} & k_{w1} \\
      \text{ID}_{w2}  & I_{w2,p}^h & G_{w2}^h & C_{w2}^h & t_{{w2},p} & e_{{w2},p} & \alpha_{w2,p} & k_{w2}\\
      ...  & ... & ... & ... & ... & ... & ... & ...\\
      \text{ID}_{N}  & I_{N,p}^h & G_{N}^h & C_N^h & t_{{N},p} & e_{{N},p} & \alpha_{N,p} & k_{N}
    \end{bmatrix}    \nonumber
\end{equation}
The $\text{ID}_{c1}, I_{w1,p}^h, G_{w1}^h, C_{w1}^h, t_{{w1},p}, e_{{w1},p}, \alpha_{w1,p}, k_{w1}$ are workshop ID, equipment options, manufacturing output, manufacturing input, manufacturing time, manufacturing energy consumption, yield rate, and geographic location, respectively. The core variable characterizing the production state is $I_{n,p}^{h}$. It is a 0-1 variable that denotes the state of the $p$-th equipment or service for the workshop $n$ at the moment $h$. The above production-related variables need to satisfy the following constraints.
\begin{subequations}
\begin{gather}
\sum\nolimits_{p=1}^{P}I^{h}_{n,p} \le 1   \label{4a}\\
\sum\nolimits_{h=1}^{H}\sum\nolimits_{p=1}^{P}I^{h}_{n,p}\le N_{n}   \label{4b}\\
T^h = \sum\nolimits_{n=1}^{N}\sum\nolimits_{p=1}^{P}I^{h}_{n,p} \cdot t_{n,p}   \label{4c}\\
E^h = \sum\nolimits_{n=1}^{N}\sum\nolimits_{p=1}^{P}I^{h}_{n,p} \cdot e_{n,p}   \label{4d}\\
\tilde{T}^h=T^h+\frac{B_m^h-B_M^{h-1}}{k_m} \cdot \Delta t   \label{4e}\\
G^h_n=\sum\nolimits_{p=0}^P I_{n,p}^h \cdot g_{n,p} \cdot \alpha^h_{n,p}   \label{4f}\\
C^h_n=\sum\nolimits_{p=0}^P I_{n,p}^h \cdot c_{n,p}   \label{4g}
\end{gather}        
\end{subequations}

Constraint (\ref{4a}) represents the uniqueness of the service for each workshop at the same moment. (\ref{4b}) expresses the maximum usage limit for each piece of equipment per day. (\ref{4c}) and (\ref{4d}) characterize the production time and energy cost at time $h$. (\ref{4e}) indicates the production time cost includes the cost of transportation time cost and equipment operation time. (\ref{4f}) and (\ref{4g}) express the production output and input for workshop $n$ at time $h$. 

\begin{remark}
    (\textit{DDU variable $\alpha^h_{n,p}$}) Production yield is a decision-dependent uncertainty (DDU). It is influenced by the inherent uncertainties of the equipment itself. Additionally, the yield of specific equipment is affected by upstream processes or the combination of production lines, making it decision-dependent. For example, in engine assembly, semi-finished products that have undergone the crankshaft grinding process have a higher success rate during the crankshaft mounting process.
\end{remark}

The time-rolling base that has been adopted in the existing work is not compatible with the actual production process. In the context of discrete manufacturing dispatch, factories are treated as discrete-time systems. The time-rolling base postulates that at each moment, the operator is capable of ascertaining the production status of the equipment. This approach does not account for the duration of equipment utilization. The duration of a single-use varies from one piece of equipment to another. It assumes that the production states for certain equipment of neighboring moments $I_{n,p}^{h}$ and $I_{n,p}^{h+1}$ are independent. However, in practice, large machines with longer minimum-use times are costly to start up, indicating that the production states of neighboring moments are not independent. Figure \ref{fig2} illustrates this: if heavy equipment is started up at moment t, it will still run at moment t+1 for economic reasons. So we design a new constraint as follows. It indicates the single minimum usage time $T_{n,p}$ for equipment $p$.
\begin{align}
    \sum_{k \in P(n)}I_{n,k}^h \cdot (I_{n,k}^{h+t}-I_{n,p}) , \forall t = 0,1,...,T_{n,p}
\end{align}

\begin{figure}[!t]
\centering
\includegraphics[width=3in]{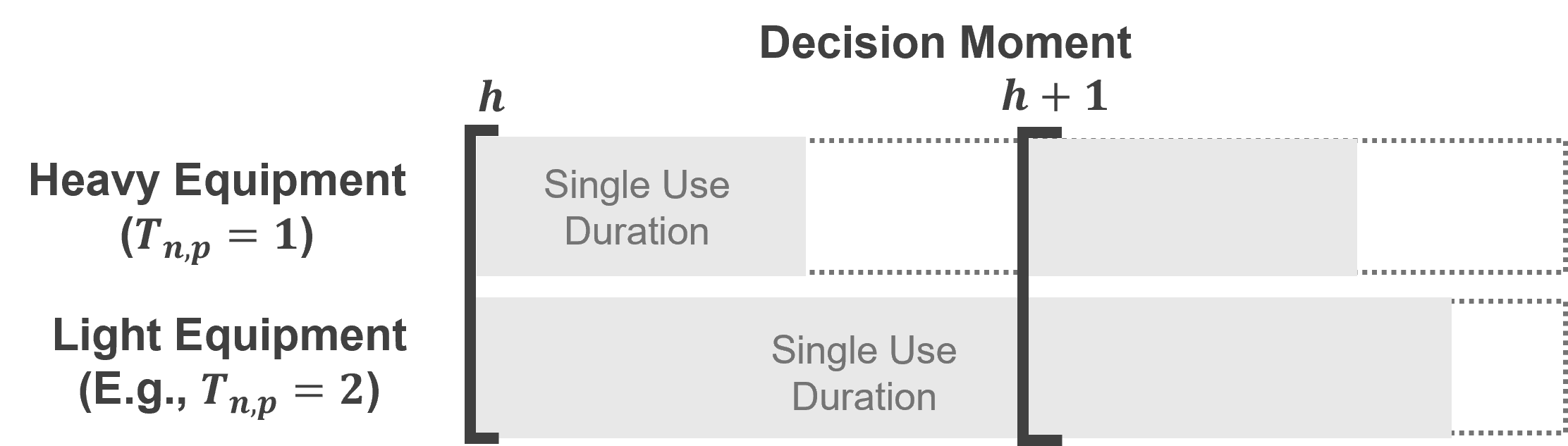}
\caption{Difference in single-use time for different equipment.}
\label{fig2}
\end{figure}
We have just discussed constraints related to the production line, and next, we will analyze the warehousing and logistics constraints. For this part, buffers exist at the back end of each workshop for storage and distribution of semi-finished products. We obtain the following constraints. 
\begin{subequations}
\begin{gather}
B_m^{h+1} = B_m^h +\Delta W_m^{h+1},  m\notin M_s   \label{6a}\\
B_m^{h+1} = B_m^h -B_{ss}^h +\Delta W_m^{h+1},  m\in M_s   \label{6b}\\
\Delta W_m^{h+1}=\sum\nolimits_{i \in P(m)}G_i^h-\sum\nolimits_{i \in S(m)}C_i^h   \label{6c}\\
B_m^h \ge 0   \label{6d}\\
B^0_m = \begin{cases} 
C, & \text{if } m = 0, \\
0, & \text{if } m \neq 0.
\end{cases} \label{6e}
\end{gather}        
\end{subequations}

(\ref{6a}) and (\ref{6b}) specify the time recurrence relationship for the buffers and that constraint (\ref{6b}) is for the special buffers that can directly output by-products. (\ref{6c}) represents the buffer volume of change in inventories. $W_m^{h+1}$ represents the difference between production of the upstream workshops and cost of the downstream workshops for buffer $m$. $P(\cdot)$ and $S(\cdot)$ represent the upstream and downstream workshops of the current warehouse respectively. (\ref{6d}) -(\ref{6e}) define the boundary constraints for the storage states.

\subsection{Energy Model and Constraints}
The introduction of DER devices and BESS in modern industrial parks is intended to reduce the cost of purchasing electricity. For simplicity, this paper focuses solely on solar panels as the DER device \cite{epn}. The introduction of the BESS and DER devices imposes the following constraints on the industrial park.
\begin{subequations}
\begin{gather}
\tilde{E}^h=E^h-E_{EU}^h-E_{LU}^h   \label{7a}\\
E_{external}^h=E^h_{SU}+E_{LU}^h   \label{7b}\\
\tilde{P}^h =\tilde{E}^h \cdot p^h   \label{7c}\\
S^{h+1}=S^h-\lambda  \cdot  E_{EU}^h+ \xi \cdot  E_{SU}^h   \label{7d}\\
0\le S^h \le \bar{S}   \label{7e}\\
\underline{k} \le |S^{h+1}-S^h|\le \bar{k}   \label{7f}\\
E_{fr}^h \ge \mathfrak{D} (\tilde{E}^h-E^h_{ex})   \label{7g}
\end{gather}
\end{subequations}

Equation (\ref{7a}) represents the nodal power injection constraint. (\ref{7b}) demonstrates power from the DER can be either stored in the BESS or used directly. (\ref{7c}) shows the total purchased electricity cost, where the real-time price (RTP) is in the form of peak and valley tariffs. (\ref{7d}) –(\ref{7f}) are BESS-related constraints. (\ref{7d}) is the BESS state transfer equation, and (\ref{7e}) and (\ref{7f}) are the energy storage unit capacity limit and adjacent time interval creep constraints. $\underline{k}, \bar{k}$ are lower and upper bounds. (\ref{7g}) denotes the frequency regulation penalty term, $\mathfrak{D}(\cdot)$ here represents a general distance measure function.

\begin{remark}
    (\textit{DDU variable $E_{fr}^h$}) The penalty term for frequency regulation is related to the expected power consumption $E_{ex}^h$ from the power utility, which is unknown to the factory. Meanwhile, it is also related to the actual power consumption data of the factory and uncertain variables, so it is a decision-related uncertainty.
\end{remark}

\section{Modeling of DDUs} \label{Section3}
This section will look at the three most prominent forms of DDU in actual discrete manufacturing production. The mathematical forms of these three DDU variables are given and then simplified with different linearization methods. 

\subsection{Classification of DDUs}
For DDU variables, we cannot grasp their specific values but can estimate them. It can be divided into two forms according to the variable's characteristics and state. If the decision-maker can estimate the distribution of DDU through historical data or physical law, then the probability distribution can be used to characterize DDU. If the decision-maker knows nothing, the DDU should be described by ambiguous sets. For probability DDU, it can also be divided into univariate probability DDU and multivariate probability one.

Returning to the actual production process of discrete manufacturing, we will now classify the above-mentioned DDUs as follows.

\begin{itemize}
    \item \textbf{Product yield DDU}: For product yield DDU, the yield is related to the decision variable of different production line combinations, but the decision maker cannot estimate the uncertainty of the equipment in each combination. Therefore, it can only be described by the rough boundary, which belongs to the ambiguous set DDU.
    \item \textbf{Frequency regulation penalty DDU}: The DDU of the frequency regulation penalty function is related to the frequency regulation demand from the power utility, which is unknown to the factory. Meanwhile, decision-makers can estimate this expectation value from historical data.
    \item \textbf{Product Structure DDU}: Decision-makers can estimate the tendencies for different products based on historical order information. However, the historical information includes many factors such as production line combination and specific equipment selection, so this DDU needs to use a multivariate distribution model. Without loss of generality, the imprecise Dirichlet model is selected in this paper.
\end{itemize}

\subsection{Definition and Simplification of DDUs}
\subsubsection{\textbf{Product Yield DDU}}
As mentioned above, the combination of production lines will affect the yield of a certain product. For example, if a specific combination of production lines includes additional overhaul processes for the same equipment, then product yields will increase. We redefine the $\alpha^h_{n,p}$ as follows.
\begin{subequations}
\begin{gather}
\alpha^h_{n,p} \ge \underline{\alpha}_{n,p} -  d_{n,p}(I^h_{n,p}), p \in \mathcal{P}   \label{8a}\\
d^h_{n,p}=\sum\nolimits_{k \in \mathcal{K}}I^{k,h}\Delta\alpha_k   \label{8b}\\
I^{k,h}= \prod\nolimits_{k_i\in k}I^h_{n,k_{1}}I^h_{n,k_{2}}...I^h_{n,k_{N}}   \label{8c}\\
G^h_n=\sum\nolimits_{p=0}^P I_{n,p}^h g_{n,p}\alpha_{n,p}^h   \label{8d}
\end{gather}
\end{subequations}

In (\ref{8a}), the correction parameter $d_{n,p}$ is denoted as the combination of production lines, see (\ref{8b}). $\underline{\alpha}_{n,p}$ is the set lower bound. The set $\mathcal{P}$ is the set of equipment to be corrected, and the set $\mathcal{K}$ is the combination of production lines that will affect the finished product rate.

\subsubsection{\textbf{Frequency Regulation Penalty DDU}}
For $E_{ex}^h$ in the frequency regulation term of the objective function, we will first denote the general frequency regulation distance as follows. $\mathfrak{D}$ denotes the general distance measure function, and we use L-1 norm $\mathcal{D}$ here.
\begin{equation}
E_{fr}^h \ge \mathcal{D} (\tilde{E}^h-E^h_{ex})
\end{equation}

The historical expectation data is represented by the Gaussian distribution as (\ref{10a})-(\ref{10c}). Suppose the dataset for time $h$ contains $T$ samples.
\begin{subequations}
\begin{gather}
\mu_h=\frac{1}{T}\sum\nolimits_{n=1}^T E_{ex,n}^h   \label{10a}\\
\sigma_h^2=\frac{1}{T}\sum\nolimits_{n=1}^T (E_{exp,n}^h-\mu_h)^2   \label{10b}\\
\mathcal{R}_1 =  \begin{Bmatrix} 
  &  \mathbb{E}_{\mathbb{P}_r}[E^h_{fr}]=\mu_h \\
 \mathbb{P}_r \in \mathcal{P}(\mathbb{R}): & \\
  & \mathbb{E}_{\mathbb{P}_r}[(E^h_{fr}-\mu^h)^2]=\sigma^2_h
\end{Bmatrix}   \label{10c}
\end{gather}
\end{subequations}

In the above formula, $\mu_h$ and $\sigma_h$ represent the mean and standard deviation. For the sake of reducing penalties, we hope the actual electricity consumption $\tilde{E}^h$ also satisfies the Gaussian distribution but includes errors due to power consumption fluctuation. We're going to modify the Gaussian distribution as follows. 
\begin{subequations} \label{11}
\begin{gather}
\mu_h(\tilde{E}^h) = \mu_h-(K\cdot \tilde {E}^h+B)\\
\mathcal{R}_2 =  \begin{Bmatrix}
  &  (\mathbb{E}_{\mathbb{P}_r}[E^h_{fr}]-\mu_h(\tilde{E}^h))^2 \cdot \sigma_h^{-1} \le\gamma_1 \\
 \mathbb{P}_r \in \mathcal{P}(\mathbb{R}): & \\
  & \mathbb{E}_{\mathbb{P}_r}[(E^h_{fr}-\mu^h)^2] \le \gamma_2\cdot \sigma_h
\end{Bmatrix}
\end{gather}
\end{subequations}

The probability distribution inscription has been revised as (\ref{11}). After obtaining $\mathcal{R}_2$, a linearization step is needed for simplification. We first rewrite the constraint in the form (\ref{12}).
\begin{equation}\label{12}
    \inf_{\mathbb{P}_r \in \mathcal{R}_2} \mathbb{P}_r(E^h_{fr} \le C) \ge 1-\epsilon 
\end{equation}

After that, based on the methodology in \cite{abs}, we first introduce variables $\tilde{s} = E^h_{fr}-\mu_h(\tilde{E}^h)$ and $b = C-\mu_h(\tilde{E}^h)$, and employ the corresponding sets.
\begin{subequations}
\begin{gather}
S=\{ (\mu_1,\sigma_1):|\mu_1| \le \sqrt[]{\gamma_1\sigma_h},\mu_1^2+\sigma_1^2 \le \gamma_2\sigma _h  \}
\\
\mathcal{R}_{2}=\begin{Bmatrix}
  &  |\mathbb{E}_{\mathbb{P}_r}[\tilde{s}]| \le \sqrt{\gamma_1}\sqrt{\sigma_h}  \\
 \mathbb{P}_r \in \mathcal{P}(\mathbb{R}): & \\
  & \mathbb{E}_{\mathbb{P}_r}[\tilde{s}^2] \le \gamma_2\cdot \sigma_h
\end{Bmatrix}
\end{gather}
\end{subequations}

Then it follows that.
\begin{multline}\label{14}
\inf_{\mathbb{P}_r \in \mathcal{R}_2} \mathbb{P}_r(E^h_{fr} \le C) =\inf_{\mathbb{P}_r \in \mathcal{D}_{\tilde{s}}}\mathbb{P}_r\{ \tilde{s} \le b\} =\\ 
\inf_{(\mu_1, \sigma_1) \in S} \inf_{\mathbb{P}_r \in \mathcal{R}_1} \mathbb{P}_r\{\tilde{s} \le b\} 
\end{multline}

We decompose the DDU constraint into a two-layer optimization problem as (\ref{14}). For the optimization problem in the outer layer, we search for the worst-case condition, while for the inner layer, we search for the worst-case distribution bound by the given variance and mean value. Then, Cantelli's inequality is used to rewrite the constraint \cite{abs}.
\begin{equation}
\inf_{\mathbb{P}_r \in \mathcal{R}_2} \mathbb{P}_r(E^h_{fr} \le C) \le \inf_{(\mu_{1},\sigma_{1})\in S }\frac{(b-\mu_{0,h})^2}{\sigma_{0,h}^2+(b-\mu_{0,h})^2}
\end{equation}

As $\mathbb {P} _r$ for normal distribution and parameters in $S$ are known, As a result, if we present $\mathbb{P}_r $'s mean and variance as $\mu_{\mathbb{P}_r}$ and $\sigma_{\mathbb{P}_r}$, then the constraint can be changed to a linear form.
\begin{equation}\label{16}
    E^h_{fr} - \mu_{\mathbb{P}_r} - \sigma_{\mathbb{P}_r}\Phi^{-1}(\frac{(b-\mu_{0,h})^2}{\sigma_{0,h}^2+(b-\mu_{0,h})^2}) \le 0
\end{equation}

In (\ref{16}), $\Phi$ is the cumulative distribution function (CDF) of $\mathbb{P}_r$. Since $\mathcal{D}$ is the L-1 norm, we can substitute the decision variables into the final optimization constraints.
\begin{subequations}\label{17}
\begin{gather}
-\Delta E^h + E_{ex}^h \le \tilde{E}^h \le \Delta E^h + E_{ex}^h  \\
\Delta E^h = \mathcal{D}^{-1}\left[\mu_{\mathbb{P}_r} - \sigma_{\mathbb{P}_r}\Phi^{-1}(\frac{(b-\mu_{0,h})^2}{\sigma_{0,h}^2+(b-\mu_{0,h})^2})\right]
\end{gather}
\end{subequations}

In (\ref{17}), $\mathcal{D}^{-1}$ is the inverse function of the distance measure function. Since the distance measure function is known and the value in $\mathcal{D}^{-1}$ is known from historical data, (17a) is a solvable linear constraint.

\subsubsection{\textbf{Products Structure DDU}}
Since different products have different margins, the operator can make trade-offs between them. We model this tendency to produce different products as $\zeta^h$. And $\zeta^h$ is a multivariate distribution DDU. Since this is partially known to us from historical data on sales of different products, we use a probability distribution to inscribe it as follows. 
\begin{subequations}
\begin{gather}
D = \{ LC_1^h,LC_2^h,...,LC_K^h \}   \label{18a}\\
W = \{ w_1^h,w_2^h,...,w_K^h \}   \label{18b}\\
\mathbb{P}_r (w_i^h) =\mathbb{P}_r (LC_i^h)= \theta_i^h, i=1,2,...,K   \label{18c}\\
f(\theta^h)= \Gamma(s)[ \prod\nolimits_{i=1}^{n}\Gamma(s\cdot r_i)]^{-1}\prod\nolimits_{i=1}^{n}\theta_i^{s\cdot r_i-1}    \label{18d}\\
\sum\nolimits_{i=1}^nr_i=1, \quad \forall r_i  \in [0,1]    \label{18e}
\end{gather}
\end{subequations}

Suppose we get a set of historical data in a finite state space and let its corresponding probability be (\ref{18c}). Assuming that there are $K$ combinations of production line states in the past historical data, denoted by $LC_i^h$, and $w_i^h$ represents a set of by-product sales to main product ratio corresponding to the $i$-th combination of production lines (obtained from the average value of the historical data), the larger $w_i^h$ indicates the higher tendency to by-products. Then we can use the Imprecise Dirichlet model (IDM) to inscribe it as (\ref{18d}) \cite{idm}. Where $\Gamma, r_i^h, \theta_i^h, s$ represent the Gamma function, prior weight factor of the $i$-th state, probability of the $i$-th state, and equivalent sample size, which is usually considered to be 1.

According to Bayesian theory, the posterior of $\theta$ belongs to the Dirichlet distribution. Therefore, its posterior probability can be represented as (\ref{19}). 
\begin{subequations}\label{19}
\begin{gather}
N^h = n_1^h+n_2^h+...+n_K^h  \label{19a}\\
n_i = n_{i,\text{HD}}^h+n_{i,\text{RT}}^h  \label{19b}\\
f(\theta^h|N)= \Gamma(s+N^h)\left[ \prod_{i=1}^{n}\Gamma(s\cdot r_i+n_i)\right]^{-1}\prod_{i=1}^{n}\theta_i^{s\cdot r_i+n_i-1} \label{19c}
\end{gather}
\end{subequations}

In (\ref{19c}), the parameter $N$ represents the total number of observations, including the real-time production status at moment $h$ and the data at moment $h$ of history. And $m_i$ represents the number of times state $i$, $n_{i,\text{HD}}$ and $n_{i,\text{RT}}$ represent the number of times state $i$ is in the historical production process and real-time production process, respectively. Based on the above posterior density function set, then the uncertainty interval of the probability can be expressed, see (\ref{20}).
\begin{subequations}\label{20}
\begin{gather}
    \zeta^h= \sum\nolimits_{i=1}^K\theta_i^h   \label{20a} \\
    \theta _i^h = [\underline{E}(\theta_i^h),\bar{E}(\theta_i^h)]=\left[\frac{m_i}{s+K},\frac{m_i+s}{s+K}\right]   \label{20b}
\end{gather}
\end{subequations}

We need to calibrate this uncertainty set. With a certain confidence $\gamma$, we denote the intervals in (\ref{20b}) as the following upper and lower bounds as (\ref{21}).
\begin{equation}\label{21}
\begin{cases}
\underline{\theta}^h=0,\bar{\theta}^h=G^{-1}(\frac{1+\gamma}{2}),n_i^h=0 \\
\underline{\theta}^h=H^{-1}(\frac{1-\gamma}{2}),\bar{\theta}^h=G^{-1}(\frac{1+\gamma}{2}),0 \le n_i^h \le N^h \\
\underline{\theta}^h=H^{-1}(\frac{1-\gamma}{2}),\bar{\theta}^h=1, n_i^h = N^h
\end{cases}
\end{equation}

In (\ref{21}), $\gamma, H, G, I=|\underline{\theta}^h,\bar{\theta}^h|, \bar{\theta}^h and \underline{\theta}^h$ represent the confidence coefficient, CDF of $B(m_i^h,s+N^h-m_i^h)$, CDF of $B(m_i^h+s, N^h-m_i^h)$, $\gamma$ confidence bands for $\theta_i^h$ of the CDF of the real distribution, and the historical information, respectively. 

So we can write the constraints on $\zeta^h$ as (\ref{22}). More historical data will bring narrower probability intervals and make the optimization effect closer to the real \cite{idms}.
\begin{subequations}\label{22}
\begin{gather}
\zeta^h= \sum\nolimits_i^K\theta_i^h(I^h_{n,p})   \label{22a}\\
\underline{\theta}^h \le \theta _i^h(I^h_{n,p}) \le \bar{\theta}^h    \label{22b}
\end{gather}
\end{subequations}

It is worth noting that the upper and lower bounds of theta in (\ref{22b}) will be updated by the value $n_i^h$ obtained in each iteration following (\ref{21}), but it always maintains linear constraints within the same loop.

\section{Problem Transformation \& Algorithm Design} \label{Section4}
In this section, we will derive the final form of the optimization problem and divide it into the corresponding master problem (\textbf{MP}) and sub-problem (\textbf{SP}). After that, we propose a new algorithm for solving this model. We give a theoretical analysis of both the convergence and algorithmic complexity of the new algorithm.

\subsection{Problem Transformation}
Eventually, we get the final form of the optimization problem. Our optimization problem satisfies the following Assumptions 1.
\begin{multline} \label{23}
    \min_{I^h_{n,p}}\sum\nolimits_{h=0}^H\tilde {T}^h+f_B(S^h) + \sum\nolimits_{h=0}^H\tilde{P}^h
    \\+ \max_{\text{DDU,DIU}} \min_{E^h_{EU},E^h_{LU}} B_M^H s_m-\zeta^h \sum\nolimits_{h=0}^HB^h_{ss}s_s \\
    + \sum\nolimits_{h=0}^H E_{fr}^h
\end{multline}
\begin{equation*}
   s.t. (2)-(7),(8a),(17),(22)
\end{equation*}

\begin{definition}\label{ass1}
    The objective function and constraints in the \textbf{MP} and \textbf{SP} of the two-stage problem are in linear form. The variable types are 0-1 integer variables or continuous variables in $\mathbb{R}_+$.
\end{definition}

Our problem setup satisfies Definition \ref{ass1}. We can divide the problem (\ref{23}) into the master problem (\textbf{MP}) and a sub-problem (\textbf{SP}). The decision variables of \textbf{MP} are $X:= I^h_{n,p}$, and the decision variables of \textbf{SP} are $Y:=\left( E_{EU}^h,E_{LU}^h \right)$ and related variables. Such a partition is reasonable, the variables for \textbf{MP} are the production decision variables while the \textbf{SP} parts are energy dispatch variables. 
\begin{align*}
    \textbf{\text{MP}}:\min_{I^h_{n,p},\psi} & \sum\nolimits_{h=0}^H\tilde {T}^h + \psi \\
    \text{s.t.} \quad & (2)-(6) & \\
    & \psi \ge \textbf{\text{SP}}(I^h_{n,p}) & \\
    & (8a): \underline{\alpha}_{n,p} -  d_{n,p}(I^h_{n,p}) \le \alpha^h_{n,p}, \quad p \in \mathcal{P}  & \\
    & (22a): \zeta^h= {\textstyle \sum_{i}^{M}} \theta_i^h(I^h_{n,p}) & 
\end{align*}
\begin{align*}
    \quad \textbf{\text{SP}}(I^h_{n,p}): \max_{\text{DDU, DIU}} & \min_{E^h_{EU},E^h_{LU}}   \sum\nolimits_{h=0}^H E_{fr}^h \\
    & +f_B(S^h) + \sum\nolimits_{h=0}^H\tilde{P}^h \\
    & \qquad +B_M^H\cdot s_m-\zeta^h \sum\nolimits_{h=0}^HB^h_{ss}\cdot  s_s \\
     \text{s.t.} \quad (7)\qquad \quad & \\
     (17a): -\Delta & E + E_{ex}^h \le \tilde{E}^h \le \Delta E^h + E_{ex}^h &
\end{align*}

The above problem is a two-stage robust problem with DDUs and the general form of such problem can be summarized as follows.
\begin{gather}\label{24}
    \Upsilon \triangleq \{X, Y, W, f, G\}
\end{gather}

In (\ref{24}) $X$ represents the \textbf{MP} decision variable, $Y$ represents the \textbf{SP} decision variable, $W$ represents the uncertainty variable, $f$ is the objective function, and $G$ is a mapping function between decision variable sets and decision-dependent uncertainty set. Furthermore, in our decision-dependent two-stage robust optimization problem, we assume our model satisfies the following Assumption. 

\begin{definition}\label{ass2}
    (\textit{DDU Independence}) For a deterministic DDU variable, its decision-relevant part is designed with only one decision variable, i.e., $X$ or $Y$.
\end{definition}

\begin{definition}\label{ass3}
    (\textit{DDU Separability}) Without loss of generality, the DDU set is said to be separable if it can be expressed as $W(x)=\{u=C(\xi,x) | \xi \in \Xi \}, \forall x \in X$. Where $\xi$ is an auxiliary random variable, and its uncertainty set $\Xi$ is irrelevant to the decision-making process and is a support set that is independent of the decision variable. $C(\cdot): \Xi \times X \rightarrow U$ is a coupling function with $\xi$ and $x$ as input parameters and $U$ as value range.
\end{definition}

Our problem setup satisfies Definition \ref{ass2} and Definition \ref{ass3}. According to Definition \ref{ass2}, the set of DDUs $W$ can be categorized into $W_X$ and $W_Y$. Furthermore, in accordance with Definition \ref{ass3}, both $W_X$ and $W_Y$ can be partitioned into decision-relevant deterministic expressions and DIU parts. We rewrite the problem $\Upsilon$ as follows. The DDU variable separation in Definition \ref{ass2} can be implemented using the linearization method described in Section \ref{Section3}.
\begin{gather}\label{25}
    \Upsilon^{\text{DDU}} \triangleq \{X, Y, W_X, W_Y, f, G_x, G_y \}
\end{gather}

\begin{algorithm}[!t]
	\caption{DDCCG Algorithm}
	\begin{algorithmic}[1]
		\STATE Set \textit{LB}=$-\infty$, \textit{UB}=$+\infty$, k=0, $\mathcal{O}=\emptyset$ \\
        \STATE Solve the master problem.
               \begin{flalign}
               &\textbf{\text{MP}}:\min_{I^h_{n,p},\psi}\sum\nolimits_{h=0}^H\tilde {T}^h+\psi \nonumber &&\\
               &s.t. \quad (2)-(6) \nonumber \\
               & \qquad \psi \ge \textbf{\text{SP}}(I^h_{n,p},E^h_{EU,l},E^h_{LU,l}), \quad \forall l \in \mathcal{O} \nonumber \\
               & \qquad \underline{\alpha}_{n,p} - d_{n,p}(I^h_{n,p}) \le \alpha^h_{n,p,l}, \quad p \in P,  \forall l \le k  \nonumber \\
               & \qquad \zeta^h_l= {\textstyle \sum_{i}^{M}} \theta_i^h(I^h_{n,p}), \quad \forall l \le k \nonumber \\
               & \qquad \tilde{E}^{h}=E^h \cdot e_{n,p}-E_{EU,l}^{h}-E_{LU,l}^{h}, \forall l \le k \nonumber
               \end{flalign}
               Derive an optimal solution $(I^{h*}_{n,p,k+1},\psi^*_{k+1})$. \\
               Update \textit{LB}= $\textbf{\text{MP}}(I^{h*}_{n,p,k+1},\psi^*_{k+1})$. \\
        \STATE Equate the $W_X$ type DDU constraint with $I^{h*}_{n,p,k+1}$.
                \begin{flalign}
                    &\underline{\alpha}_{n,p} - d_{n,p}(I^{h*}_{n,p,k+1}) = \alpha^h_{n,p,k+1}, \quad p \in P \nonumber &&\\
                    &\zeta^h_{k+1}= {\textstyle \sum_{i}^{M}} \theta_i^h(I^{h*}_{n,p,k+1}) \nonumber
                \end{flalign}
        \STATE Call the oracle to solve $\textbf{\text{SP}}(I^{h*}_{n,p,k+1}, \alpha^h_{n,p,k+1}, \zeta^h_{k+1})$
                \begin{flalign}
                    & \mathop{\max}\nolimits_{E_{fr}^h} \mathop{\min} \nolimits_{E^h_{EU},E^h_{LU}} \sum\nolimits_{h=0}^H E_{fr}^h \nonumber && \\
                    & \qquad\qquad\qquad+f_B(S^h) +\sum\nolimits_{h=0}^H\tilde{P}^h \nonumber\\
                    & \qquad\qquad\qquad\qquad\qquad+s_mB_M^H -\zeta^h \sum\nolimits_{h=0}^Hs_sB^h_{ss} \nonumber\\
                    & s.t. \quad (7a)-(7g), (17a) \nonumber\\
                    & \qquad E^h = \sum\nolimits_{n=1}^{N}\sum\nolimits_{p=1}^{P}I^{h*}_{n,p,k+1} \cdot e_{n,p} \nonumber\\
                    & \qquad\tilde{E}^h - E(E^h_{EU},E^h_{LU}) \le E^h_{ex} \nonumber
                \end{flalign}
               Derive $\textbf{\text{SP}}^* = \textbf{\text{SP}}^*(I^{h*}_{n,p,k+1}, \alpha^h_{n,p,k+1}, \zeta^h_{k+1})$ \\
               Update \textit{UB}=$\min\{\text{\textit{UB}}, \text{\textit{LB}}- \psi^*_{k+1} + \textbf{\text{SP}}^*\}$. \\
        \STATE if \textit{UB}-\textit{LB}$\le \epsilon$, return $x^*_{k+1}$ and terminate. Otherwise, do: \\
        (a) if $\textbf{\text{SP}}^* \le +\infty$, then: \\
        \qquad Create $E_{EU,k+1}^{h}, E_{LU,k+1}^{h}$, add following constraints.
        \begin{flalign}
            & \qquad \psi \ge \textbf{\text{SP}}(I^h_{n,p},E^h_{EU,k+1},E^h_{LU,k+1}) \nonumber && \\
            & \qquad \tilde{E}^{h}=E^h-E_{EU,k+1}^{h}-E_{LU,k+1}^{h} \nonumber
        \end{flalign}
        \qquad to \textbf{MP} and update $k=k+1,\mathcal{O}=\mathcal{O} \cup \{k+1\}$ and go to Step 2.\\
        (b) if $\textbf{\text{SP}}^* =  +\infty$, then: \\
        \qquad Create $E_{EU,k+1}^{h}, E_{LU,k+1}^{h}$, add following constraints.
        \begin{flalign}
            & \qquad \tilde{E}^{h}=E^h-E_{EU,k+1}^{h}-E_{LU,k+1}^{h} \nonumber && 
        \end{flalign}
        \qquad to \textbf{MP} and update $k=k+1$ and go to Step 2. \\
	\end{algorithmic}  
        \label{algo2}
\end{algorithm}

In (\ref{25}), we also introduce an updated version of the mapping function $\{G_x\colon X\times Y\rightarrow W_X\}$ and $\{G_y\colon Y\rightarrow W_Y\}$. In our work, $I^h_{n,p}$ is the \textbf{MP} decision variable $X$, while $E^h_{EU}$ and $E^h_{LU}$ are \textbf{SP} decision variable $Y$s. The DDU variable $E_{fr}^h$ is associated with $E^h_{EU}$ and $E^h_{LU}$, while DDUs $(\alpha^h_{n,p}, \zeta^h)$ are related to $I^h_{n,p}$. Therefore, $(\alpha^h_{n,p}, \zeta^h)$ belongs to $W_X$ and $E_{fr}^h$ belongs to $W_Y$. 

After linearization, we obtain the mixed integer programming problem of (\ref{23}) (Definition \ref{ass1}), where the objective functions and constraints in \textbf{MP} and \textbf{SP} are linear in form. Since $Y$ is linear concerning $X$ and $Y$ and $G_y$ is a one-to-one mapping, we consider $Y$ as a mapping for $X$ and $W_Y$.

At this juncture, if $W_X$ is an empty set, the problem degenerates into a normal two-stage robust optimization problem, which can be solved using the traditional algorithm. Among the traditional algorithms, the C\&CG algorithm is selected for solving the problem due to its rapid convergence rate and adaptability to the requirements of modern industry. However, given that $W_x$ is not an empty set and traditional algorithms such as C\&CG are not suitable for $W_X$-type DDU optimization problems, the development of new algorithms is necessary. To this end, we propose an \textbf{D}ecision-\textbf{D}ependent \textbf{C}olumn-and-\textbf{C}onstraint \textbf{G}eneration (\textbf{DDCCG}) algorithm. In our proposed algorithm, an additional step is designed in our new algorithm to gradually generate constraints on $W_X$. These constraints are generated by adding Benders cutting planes. The analysis of the convergence and optimality of the algorithm is presented in the next subsection.

The DDCCG algorithm is an iterative algorithm that solves decision-dependent two-stage robust optimization problems. The core idea of the algorithm is to gradually generate and add constraints to approach the optimal solution.

\begin{itemize}
    \item \textbf{Step 1}: The algorithm starts by initializing the parameters.
    \item \textbf{Step 2}: Step 2 first obtains the decision variables of the first stage.
    \item \textbf{Step 3}: Step 3 Based on the $X$ derived from Step 1, we tighten $W_X$ type DDU constraint, obtain the relaxed $W_X$ variable value.
    \item \textbf{Step 4}: Subsequently, Step 4 solves the sub-problem based on the obtained $X$ and $W_X$ to obtain the current optimal $W_Y$ and $Y$.
    \item \textbf{Step 5}: Step 5: The iterative process continues until the termination condition is met. Otherwise, similar to the C\&CG algorithm, variables (columns) and constraints are generated and added to the \textbf{MP}.
\end{itemize}

\subsection{Algorithm Performance Analysis}
Under Definition \ref{ass2} and Definition \ref{ass3}, our co-scheduling problem is a two-stage robust optimization model with an additional $W_X$-type DDU computation session. This section will analyze the convergence and algorithmic complexity of the DDCCG algorithm.

\subsubsection{Convergence}
For the convergence of this problem, we give the following Proposition \ref{prop1}.

\begin{proposition}\label{prop1}
    If a decision-dependent two-stage robust optimization problem with both $W_X$ and $W_Y$ types of DDU variables satisfies Definition \ref{ass1}, the DDCCG algorithm can converge in finite steps.
\end{proposition}

\begin{proof}
    In each step, the feasible set of $Y$ is under the influence of $X^*_{k+1}$ and $W_{X,k+1}$ derived from Step 2-3. Under Definition \ref{ass2} and Definition \ref{ass3}, the feasible set of the \textbf{MP} variable $X$ is a finite number $N_x$ of discrete points in a polyhedron, and thus the feasible set of the \textbf{SP} variable $Y$ is a set of polyhedra. Since the feasible solutions of $X$ are finite, the feasible set of $Y$ in the optimal condition is a single polyhedron.

It thus suffices to demonstrate whether the optimal solution of $Y$ considering WY for a decision-dependent problem can be solved in finite steps. That is to say, i.e., to prove the convergence of Steps 4-5.

In our model, the $W_Y$ type DDU constraints are expressed as follows.
\begin{flalign}
(17a): -\Delta E^h + E_{ex}^h \le \tilde{E}^h \le \Delta E^h + E_{ex}^h \nonumber
\end{flalign}

Under Definition \ref{ass1}, We have $N_x$ polyhedra forming the feasible set for $Y$. The definition of a single polyhedron is (17a). For each polyhedron, we denote all the extreme points $W_Y^P \triangleq \{ w_{Y1}^P,...,w_{YI}^P \}$ and extreme rays $W_Y^R \triangleq \{ w_{Y1}^R,...,w_{YJ}^R \}$ for $W_Y=E_{fr}^h$. $I$ and $J$ are finite.

Extreme rays and extreme points correspond to the requirements for solvability and optimality of the problem respectively \cite{benders}.
\begin{subequations}
\begin{gather}
(E_{fr}^h)^R_j \ge \mathcal{D} (\tilde{E}^h-E^h_{ex}), \quad \forall j =1,...,J\\
\psi \ge \textbf{\text{SP}}(I^{h*}_{n,p}, (E_{fr}^h)^P_i), \quad \forall i =1,...I
\end{gather}
\end{subequations}

Our algorithm, in the \textbf{MP} stage, produces a set of candidate optimal solutions $(I^{h*}_{n,p}, \alpha^{h*}_{n,p}, \zeta^{h*})$) and potential \textbf{SP} optimal value $\textbf{\text{SP}}^*$ represented by $\psi$, and subtitues them into \textbf{SP} problem to prodece $\textbf{\text{SP}}(I^{h*}_{n,p})$. At this point, we face the following three situations.

\begin{itemize}
    \item $\textbf{\text{SP}}(I^{h*}_{n,p}) = \textbf{\text{SP}}^*$: The algorithm stops.
    \item $\textbf{\text{SP}}(I^{h*}_{n,p}) \ne \textbf{\text{SP}}^*$:
        \begin{itemize}
            \item \textbf{$\textbf{\text{SP}}(I^{h*}_{n,p})$ is unbounded}: Additional constraints on $W_X$ type DDU are generated and added to \textbf{MP}. These kinds of constraints are called ``Benders feasibility cuts" because they enforce necessary conditions for the feasibility of the \textbf{SP}
            \item \textbf{$\textbf{\text{SP}}(I^{h*}_{n,p}) > \textbf{\text{SP}}^*$}: Then in addition to such ``Benders feasibility cuts" constraints, $\eta$ related constraints will be added, and they are called ``Benders optimality cuts" because they are based on the optimal case of \textbf{SP}.
        \end{itemize}
\end{itemize}

Since both $I$ and $J$ are finite, this algorithm must converge to an optimal solution in a finite number of steps.

The convergence process of the algorithm is shown in Figure \ref{fig3}. It can be seen that by continuously equating the DDU variable $W_X$ and adding constraints, the feasible solution boundary is continuously compressed until $X^*, W^{*}_Y, Y^*$ on the boundary coincides with the optimal point.
\end{proof}

\subsubsection{Complexity Analysis}
A more precise estimation of the algorithmic complexity will be provided herewith. First, the optimization problem in this paper satisfies Assumption \ref{ass4}.

\begin{assumption}\label{ass4}
    The decision-dependent two-stage robust optimization model has the relatively complete recourse property and the strong duality property.
\end{assumption}

Under the Assumption \ref{ass4}, the following Proposition \ref{prop2} is therefore put forth.

\begin{proposition}\label{prop2}
    Let $N_x$ be the number of discrete points in the feasible set of $X$. For each \textbf{SP} for $W_Y$ and $Y$ under a certain $X$, let $N_{W_Y}$ be the number of extreme points of $W_Y$ if it is a polyhedron. Then, the DDCCG algorithm will converge to the optimal value in $O\left(N_xN_{W_Y}\right)$ iterations.
\end{proposition}

\begin{proof}
The value of $W_X$ is related to the feasible set of X and contains the $N_x$ case. Therefore, proving Proposition \ref{prop2} is equivalent to proving that the complexity of solving the two-stage robust optimization problem without $W_X$ (remove Step 3) is $O(N_{W_Y})$.

Without loss of generality, we give the general form of a decision-dependent two-stage robust optimization problem model (containing only $W_Y$) under a certain $W_X$. We \textbf{SP}lit $W_Y$ into the expression associated with $Y$ and the DIU variable $\mathcal{U}$ according to linearization. $\mathcal{U}$ here is a polyhedron.
\begin{align}
    \min_xc^Tx & +\max_{u \in \mathcal{U}} \min_{y} b^Ty \nonumber\\
    \text{s.t.} \quad & Ax \ge d  & \nonumber\\
    & Gy \le Ex & \\
    & By \le e  & \nonumber\\
    & Hy \ge k-Mu & \nonumber \\
    & x\in S_x, y \in S_y \nonumber
\end{align}    

Under Assumption \ref{ass4}, \textbf{SP} from the above primal problem (also the \textbf{SP} in our DDCCG algorithm) can be converted into \textbf{SP} of the dualized problem as follows.
\begin{align}
    \min_yc^Ty+&\max_{u, \lambda, \gamma, \mu} e^T \gamma+(k-Mu)^T \mu+y^T E^T \lambda \nonumber\\
    \text{s.t.} \quad & Ax \ge d  & \nonumber\\
    & G^T\lambda +B^T\gamma+H^T\mu=b  & \\
    & \lambda, \gamma, \mu \ge 0 & \nonumber \\
    & u \in \mathcal{U} & \nonumber
\end{align}    

Note that \textbf{SP} is a bi-linear program over two disjoint polyhedrons, it always has an optimal solution combining extreme points of these two polyhedrons, despite the value of $x$. Let $\hat{\mathcal{U}}=\{u_1,...,u_{N_{W_Y}}\}$ be the collection of extreme points of $\mathcal{U}$. Hence, the primal two-stage robust optimization problem is equivalent to:
\begin{gather}
    \min_xc^Tx +\max_{u \in \hat{\mathcal{U}}} \min_{y} b^Ty
\end{gather}

As a result, by enumerating values in the finite set $\hat{\mathcal{U}}$, it is equivalent to a mixed integer program in the following form problem (\ref{30}) (also the \textbf{MP} in our DDCCG algorithm). Let $\hat{Y}=\{y_1,...,y_{N_{W_Y}}\}$ be the corresponding recourse decision variables.
\begin{align}\label{30}
    \min_xc^Tx & +\eta \nonumber\\
    \text{s.t.} \quad & Ax \ge d  & \nonumber\\
    & \eta \ge b^Ty_l, \quad l=1,...,N_{W_Y}  & \nonumber\\
    & Gy_l \le Ex,\quad l=1,...,N_{W_Y} & \\
    & By_l \le e,\quad l=1,...,N_{W_Y}  & \nonumber\\
    & Hy_l \ge k-Mu_l,\quad l=1,...,N_{W_Y} & \nonumber \\
    & x\in S_x, y_l \in S_y,\quad l=1,...,N_{W_Y} \nonumber
\end{align}  

Note that solving the \textbf{SP} of the primal problem will expand (\ref{30}) by including one more $u_l$ and its corresponding recourse variables and constraints. Next, we show that any repeated $u^*$ in the procedure implies optimality, i.e. \textit{LB} = \textit{UB} (in the DDCCG algorithm). Assume that in the $k$-th iteration $(x^*, \eta^*)$ and $(u^*, y^*)$ are optimal solutions of (31) and \textbf{SP} of the primal problem, respectively, and $u^*$ appear in a previous iteration. On one hand, from Step 4 of the DDCCG algorithm, we have $\textit{UB}\le c^Tx^*+b^Ty^*$. On the other hand, because $u^*$ has been identified in a previous iteration, (\ref{30}) in this iteration is identical to that in $(k-1)$-th iteration. Hence, $(x^*, \eta^*)$ is the optimal solution of (\ref{30}) in the $(k-1)$-th iteration. From Step 2 of our DDCCG algorithm, we have $\textit{LB}\ge c^Tx^*+\eta^* \ge c^Tx^*+b^Ty^*$, where the last inequality follows from the fact that $u^*$ is already identified and the related constraints are added to (31) before or in the $(k-1)$-th iteration. Consequently, we get $\textit{LB}=\textit{UB}$. Then, the conclusion follows immediately from the fact that $\hat{\mathcal{U}}$, i.e. the set of all extreme points of the polyhedral uncertainty set, is finite. Therefore, the complexity of solving the two-stage robust optimization problem without $W_X$ is $O(N_{W_Y})$.
\end{proof}

\begin{remark}
    Proposition \ref{prop1} proves that the DDCCG algorithm converges for the general decision-dependent two-stage robust optimization model. Proposition \ref{prop2} gives an estimate of the algorithm's complexity for problems that satisfy specific assumptions and structures.
\end{remark}

\begin{figure*}[!t]
\centering
\includegraphics[width=5.5in]{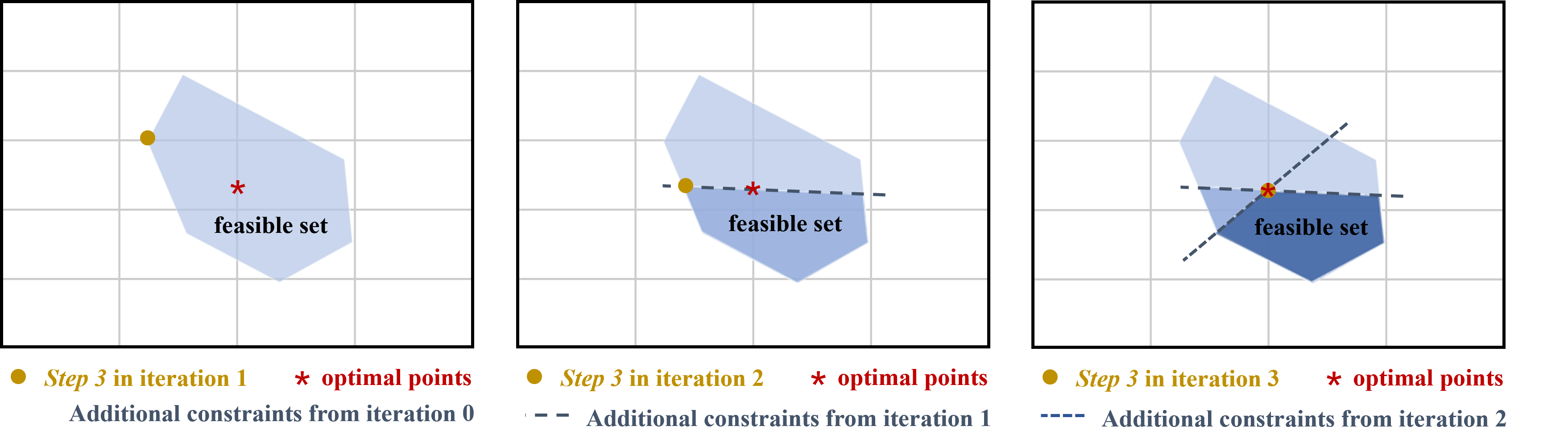}
\caption{DDCCG algorithm diagram.}
\label{fig3}
\end{figure*}

\section{Case Study} \label{Section5}
This section will initially model a real-world engine assembly factory with the Petri net. Subsequently, the model and the DDCCG algorithm will be employed to obtain simulation results, which will then be further elaborated.

\subsection{Petri Net Modeling}
To demonstrate the efficacy of our proposed methodology, we have utilized a real-world engine assembly factory as a representative case of discrete manufacturing \cite{eng}. A Petri net is employed for its quantitative description \cite{petri1}. This is because Petri net is highly effective at characterizing the sequential and concurrent relationships between events and operations over a time series \cite{petri2}. Petri net consists of two basic elements: places and transitions. Places represent states in the system, while transitions represent transfers between states. A real engine assembly factory involves multiple production, assembly, transportation, and testing stages, and we can model it using Petri net (see Fig \ref{fig4}).

From the initial raw material to the final product (engine), the entire industrial process is divided into 12 workshops, 14 types of equipment, and multiple buffers (node label in Appendix). The entire production process can be divided into 8 major segments according to the specific function of the workshop: 
\begin{itemize}
    \item \textbf{Parts Production}: stamping machine (SM)
    \item \textbf{Datum Milling}: milling machine, face milling machine
    \item \textbf{Crankshaft Grinding}: circular grinding machine (CG), CNC grinding machine
    \item \textbf{Crankshaft Mounting}: automatic assembly line ($\text{AAL}_1$), robotic assembly line ($\text{RAL}_1$)
    \item \textbf{Piston Assembly}: piston assembly machine (PAM), automatic assembly machine ($\text{AAM}_1$)
    \item \textbf{Cylinder Mounting}: automatic assembly line ($\text{AAL}_2$), robotic assembly line ($\text{RAL}_2$)
    \item \textbf{Toothed Belt Mounting}: belt mounting machine (BMM), automatic assembly machine ($\text{AAM}_2$)
    \item \textbf{Testing}: Integrated test platform
\end{itemize}. 

In real-world production, the sequence of crankshaft grinding and cylinder mounting can be altered. If the semi-finished product is processed through the same equipment a second time, the detector will perform a check and then direct it to the buffer. This process can be employed to enhance the yield. Furthermore, the P54 node has been designated as a by-product output buffer, which outputs by-products containing cylinder liners. The time frame was set at 24 hours a day.

\begin{figure*}[!t]
\centering
\includegraphics[width=6in]{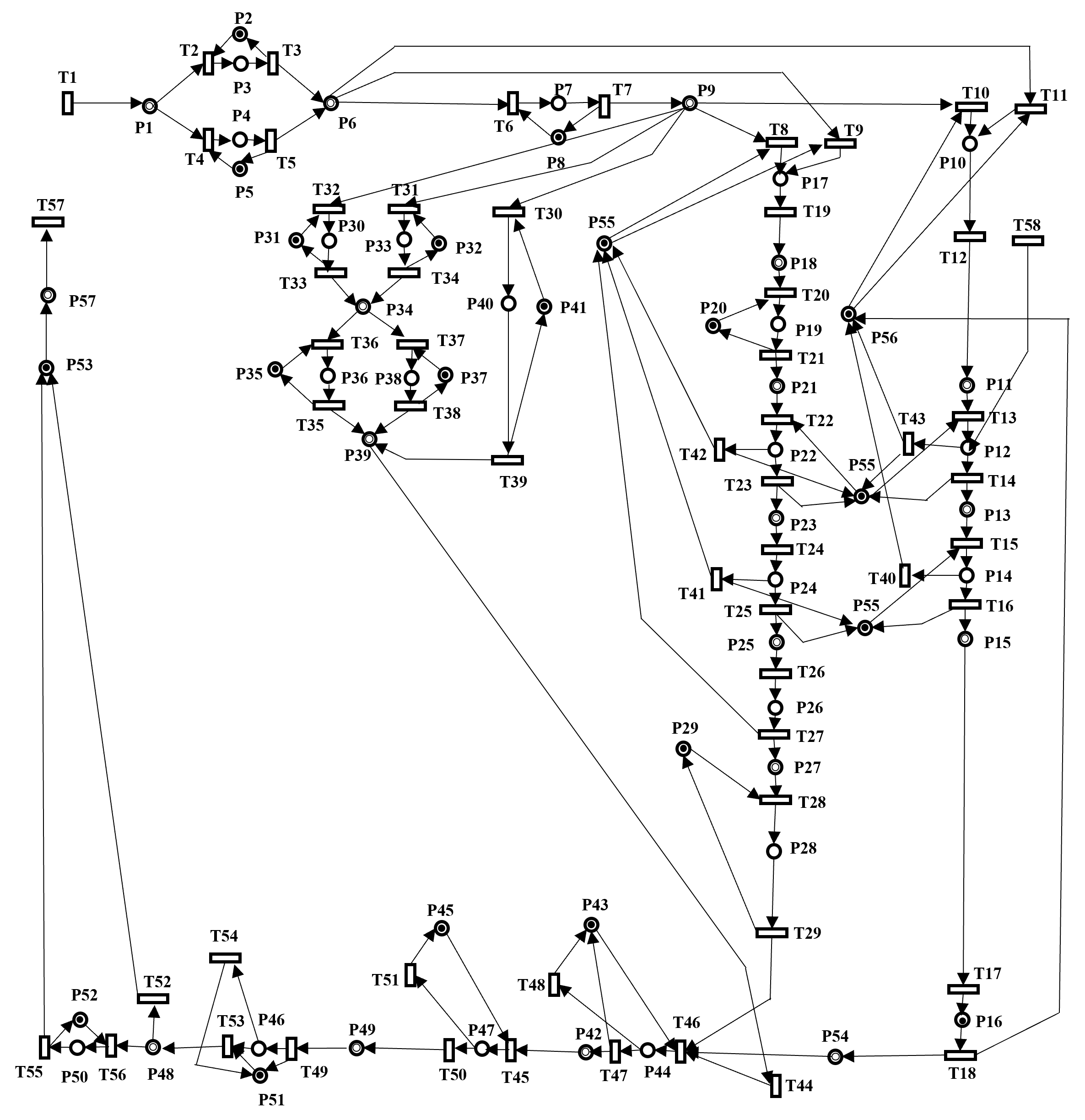}
\caption{Petri net for engine assembly lines.}
\label{fig4}
\end{figure*}

\subsection{Analysis of Results}
\subsubsection{Improved Production Performance}
The proposed model has the potential to significantly reduce production costs. As demonstrated in Table \ref{tab1}, our model, either without DDUs or with DDUs, outperforms both historical data and ideal data in terms of production performance and electricity cost. Our model permits the integrated scheduling of energy and production equipment, in contrast to a real factory. This results in a reduction in power consumption within the industrial park and an increase in production. In comparison to the model without DDUs, the model with DDUs maintains the same primary product output while reducing energy consumption and improving by-product output. We then will specifically analyze the sources of variability between the DDU model and the model without DDU.

\begin{table*}[t]
\renewcommand{\arraystretch}{1.3}
\caption{Comparison of production costs}
\centering
\begin{tabular}{cccccc}
\toprule
\textbf{ModelType} & \textbf{Power Cost (CNY)} & \textbf{Average Cost (CNY)} & \textbf{Main Products (CNY)} & \textbf{By-products (CNY)} & \textbf{Objective Function (CNY)}\\
\midrule
Original Data & 1042.00 & 43.42 & 2250 & 200 & - \\
Expected Data & 980.00 & 40.83 & 2250 & 200 & - \\
Model without DDU & 945.09 & 39.39 & 2500 & 160& -2936.41 \\
*Model with DDU & 922.97 & 38.46 & 2500 & 250 & -3340.78\\
\bottomrule
\end{tabular}
\label{tab1}
\end{table*}

The reduced electricity cost of the model with DDUs relative to the one without DDUs can be attributed to its enhanced efficacy in circumventing peak periods of electricity consumption. As illustrated in Fig. \ref{fig5}, during the RTP peak period (i.e., the period when the golden folded line reaches the peak plateau), the DDU model (dark blue) exhibits a notable reduction in electricity consumption compared to the model without DDU (light blue), with a difference of approximately 3-5\% for each segment of the bar chart.

\begin{figure}[!t]
\centering
\includegraphics[width=3in]{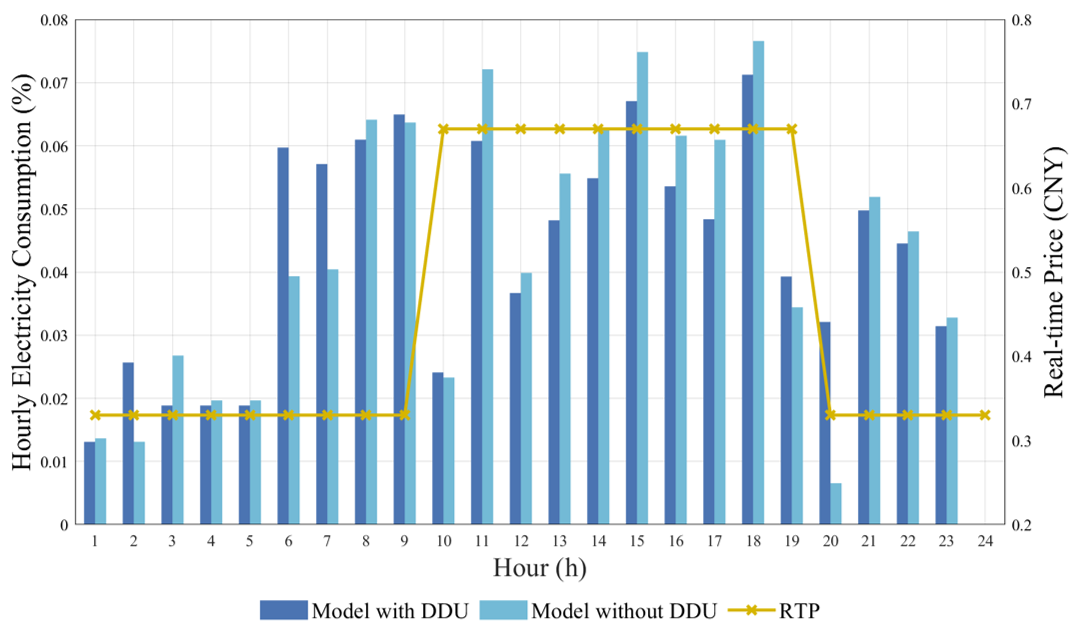}
\caption{Historiy of electricity cost and RTP for different models.}
\label{fig5}
\end{figure}

In comparison to the model without DDUs, the rise in by-product production in the one with DDUs can be attributed to the decentralization of production, which allows for the optimal utilization of redundant raw materials for the manufacture of the primary product. The total number of instances of utilization for each device is calculated over 24 hours and a matrix based on the data is plotted considering the device's location. Bilinear interpolation is then employed to expand the matrix and generate a contour map (Fig. \ref{fig6}). Compared to the model without DDUs, The DDU-containing model has two production centers of gravity (yellow areas), while the production line on the right side contains the P54 node, which can be utilized for the sale of by-products. Therefore, the utilization of redundant primary product materials for the synthesis of by-products for sale represents a viable strategy for enhancing profitability. Concurrently, the implementation of a multi-center of gravity production line can also serve to reinforce the plant's resilience to unforeseen circumstances. In the event of a disruption to one production line, the mobilization of another line to assume the role of a center of gravity can effectively mitigate the impact of such an occurrence.

\begin{figure}[!t]
\centering
\begin{subfigure}{0.25\textwidth}
  \centering
  \includegraphics[width=1.45in]{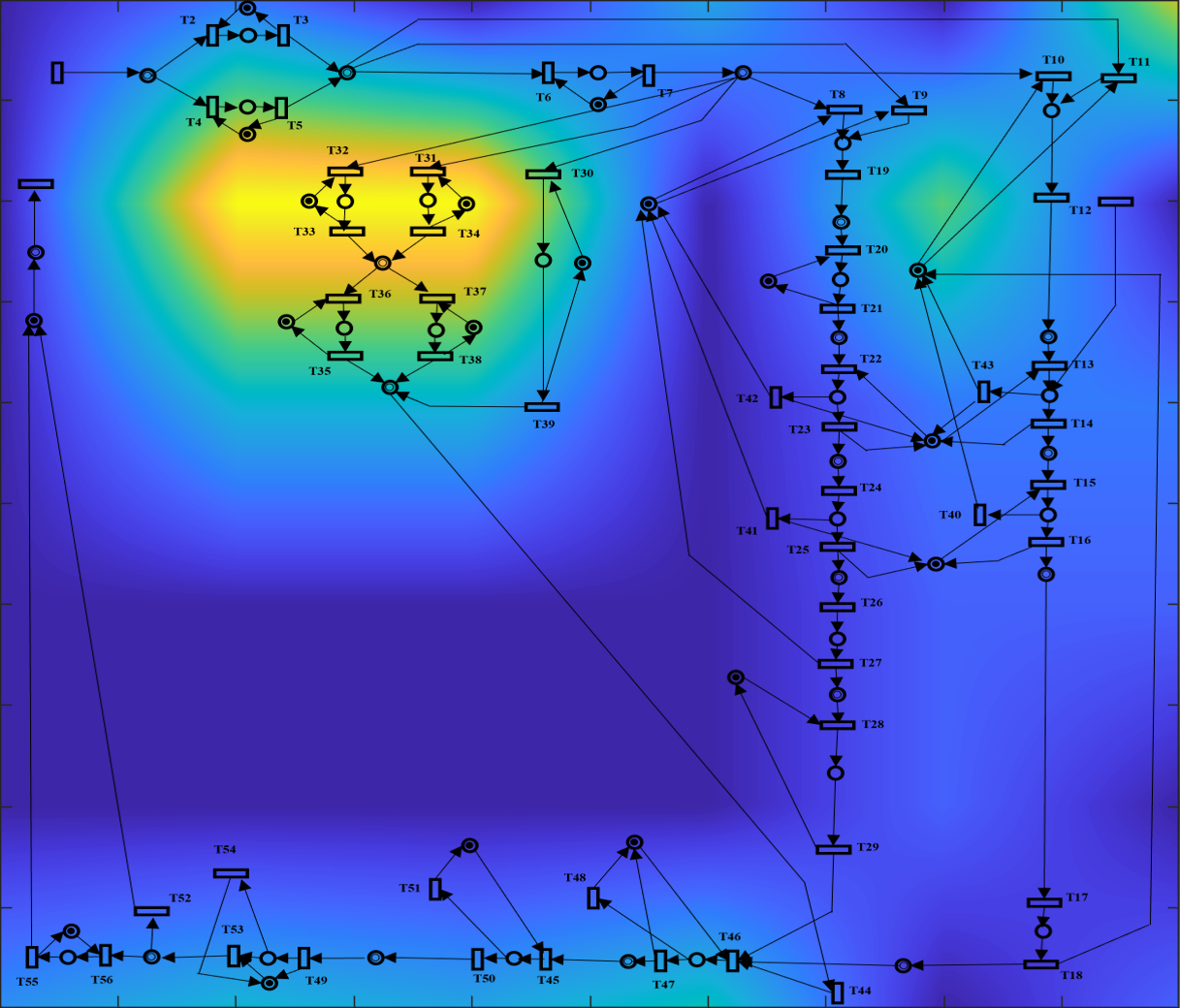}
  \caption{No-DDU Model.}
  \label{fig:sub1}
\end{subfigure}%
\begin{subfigure}{0.25\textwidth}
  \centering
  \includegraphics[width=1.45in]{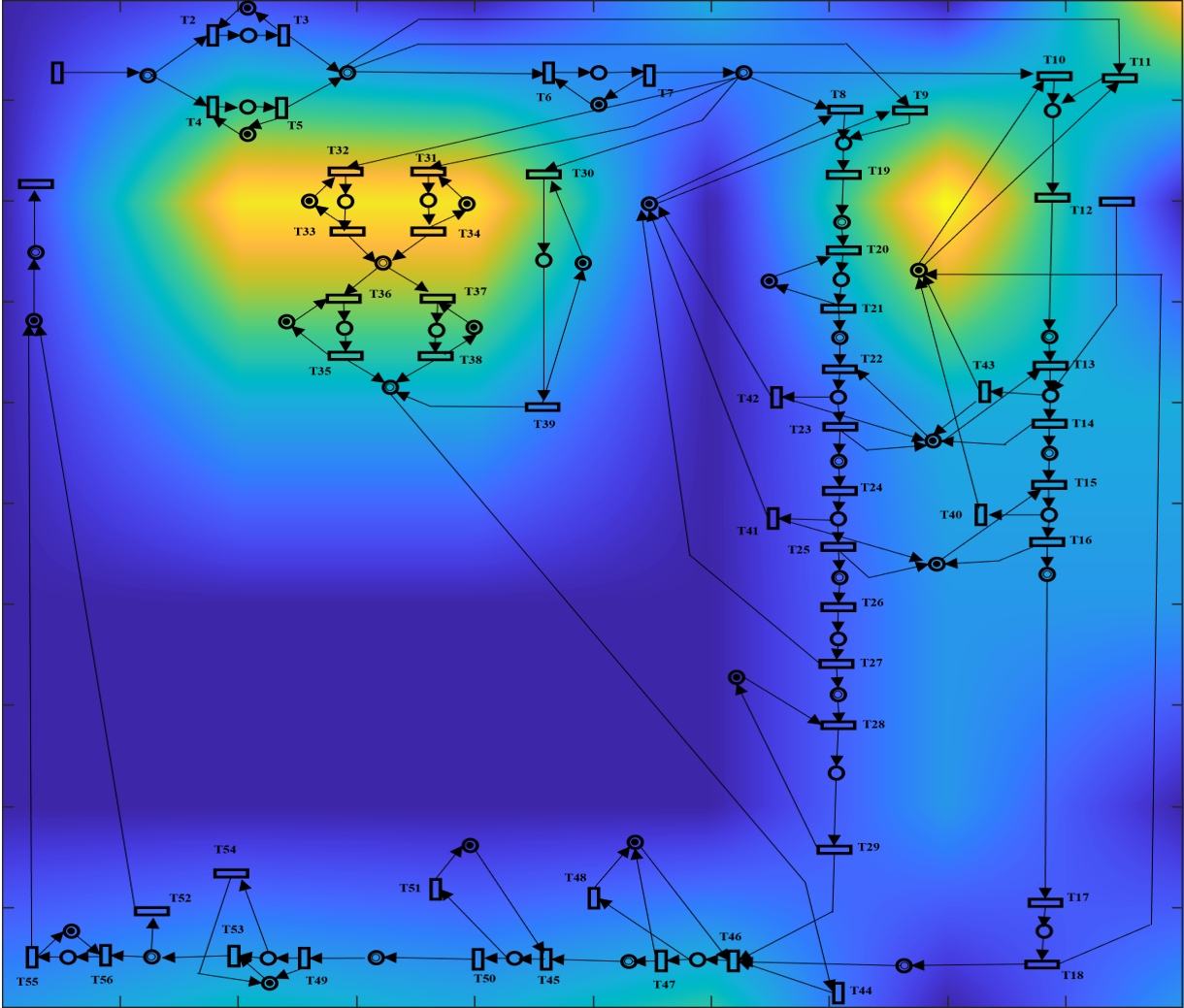}
  \caption{DDU Model.}
  \label{fig:sub2}
\end{subfigure}
\caption{Contour map of equipment utilization.}
\label{fig:test}
\label{fig6}
\end{figure}

\subsubsection{Frequency Regulation \& Peak Shaving}
The proposed model is capable of adapting to the requisite frequency regulation demands, and the DDU model can enhance the system's peak shaving capability. Fig. \ref{fig7} illustrates the outcomes of a comparative analysis between the DDUs model and the original data set, with and without the DDUs model, respectively, in response to the specified frequency regulation demand. It can be observed that both the model with DDUs and without DDUs are closer to the regulation demand fold than the original data. The difference between the model with DDUs and the model without DDUs is that it is less responsive to frequency regulation, but has a smoother power consumption fold (i.e., better peak shaving performance). These two differences will be further elucidated.

\begin{figure}[!t]
\centering
\begin{subfigure}{.5\textwidth}
  \centering
  \includegraphics[width=3.5in]{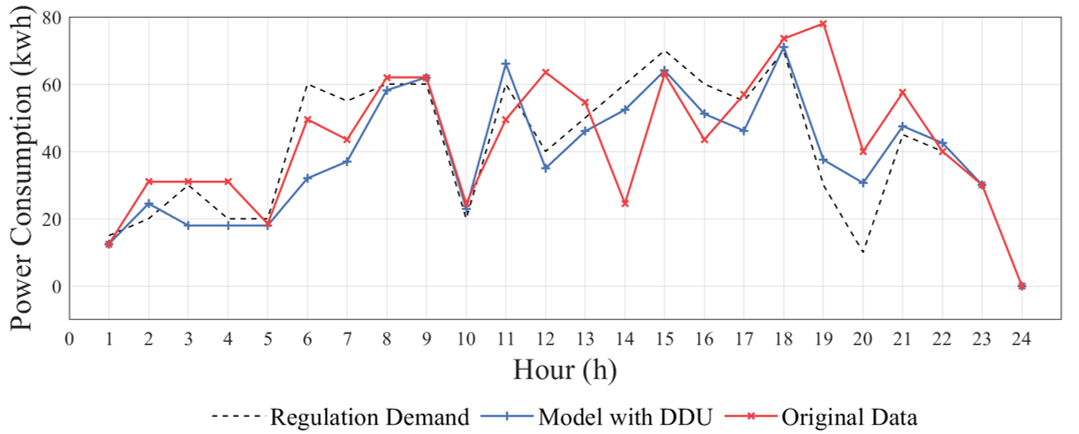}
  \caption{Comparison of DDU with raw data.}
  \label{fig2:sub1}
\end{subfigure}%
\newline
\begin{subfigure}{.5\textwidth}
  \centering
  \includegraphics[width=3.5in]{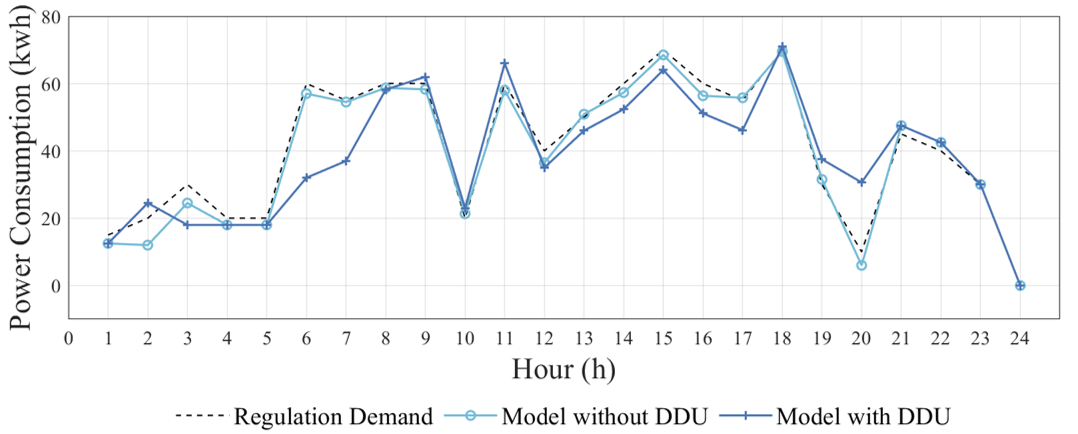}
  \caption{Comparison of DDU with No-DDU.}
  \label{fig2:sub2}
\end{subfigure}
\caption{Comparison of FR realizations for different scenarios.}
\label{fig2:test}
\label{fig7}
\end{figure}

The inferior performance of the DDU model in frequency regulation relative to the non-DDU model can be attributed to its dynamic adjustment of the weights of the frequency regulation penalty term. Fig. \ref{fig8} illustrates the values of the frequency regulation penalty term and the value of $\eta$ at different stages when we iterate the DDCCG Algorithm, while the dotted line indicates $\eta_h$  of the model without DDUs. It is evident that during the initial stage where there are more violations of the frequency regulation target, the value of $\eta$ is higher. However, once the model believes that the performance on frequency regulation has performed optimally enough, the weight is reduced, allowing the model to focus on other targets. This performance is a consequence of the DDU model prioritizing cost savings in subsequent iterations, as previously discussed. Overall, this trade-off can effectively assist producers in enhancing their profits and has practical industrial significance. The following subsection will provide a more detailed examination of this element.


The DDU model is capable of achieving superior peak shaving outcomes utilizing smoothing the production curve in comparison to the non-DDU model. The violins of electricity usage for the two models were plotted (Fig. \ref{fig9}), and it can be observed that the DDU-containing model exhibits a smoother curve and has a smaller variance. This is essentially because the DDU model contains greater uncertainty, and therefore the system adopts a conservative, smooth power usage strategy to cope with potential changes (e.g., RTP mutations). In terms of the power grid, this behavior represents increased BESS storage and more advanced power usage planning. The reduction in equipment usage during peak RTP periods is a reflection of this.

\begin{figure}[!t]
\centering
\includegraphics[width=3in]{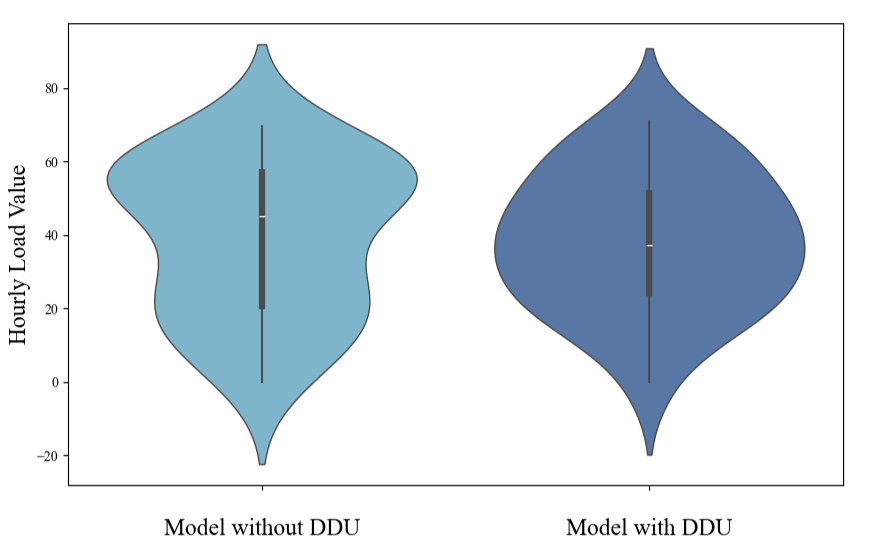}
\caption{Violin diagrams for DDU's and No-DDU's models.}
\label{fig8}
\end{figure}

\subsubsection{DDU Impact Assessment}
In the case of the yield DDU, the original value of the device's inherent error, represented by $\underline{\alpha}_{n,p}$, affects the model's estimate of the uncertainty. The $\underline{\alpha}_{n,p}$ represents the uncertainty inherent to the device's yield. Fig. 10 is constructed, plotting the model's upper bound (the optimal value for solving the master problem) with iteration for varying values of $\underline{\alpha}_{n,p}$. The error band represents the relative interval width in \textbf{MP} for $\mathcal{S}_1$. As the \textbf{SP} continues to add constraints, the error band shrinks (the operator prefers line combinations with a reduced interval). In addition, the gap did not shrink for $\underline{\alpha}_{n,p}=0.1$ in the later iterations, which is because the yield is not the only factor that affects the final profit. However, it should be noted that in \textbf{MP}, the constraints in \textbf{SP} are not necessarily satisfied for $\mathcal{S}_1$. It is evident that models with greater inherent errors, represented by $\underline{\alpha}_{n,p}$, exhibit slower convergence, inferior optimal solutions, and larger error bands. This is because the significant intrinsic yield error of the device makes the DDU estimation challenging, resulting in an increased error band. Consequently, a substantial discrepancy will result in a diminished yield, thereby compromising the optimal value.

\begin{figure}[!t]
\centering
\includegraphics[width=3in]{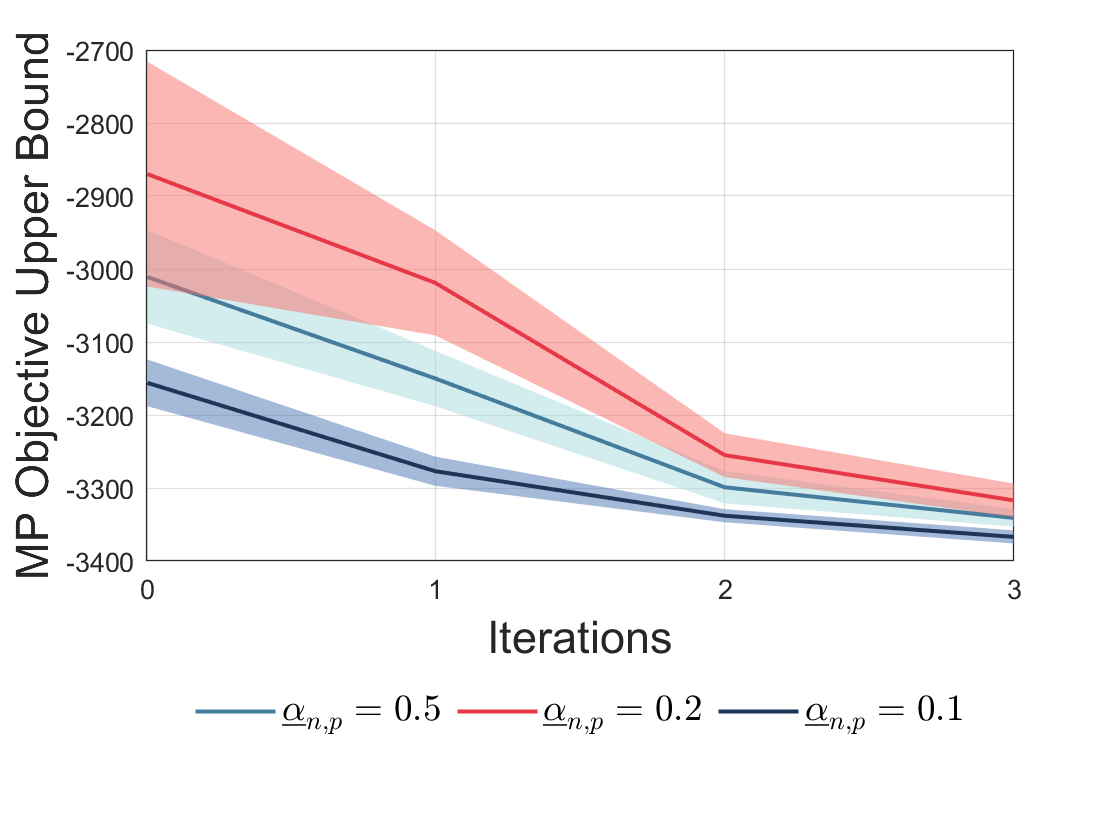}
\caption{The upper bound error band diagram under different $\underline{\alpha}_{n,p}$.}
\label{fig9}
\end{figure}

For FR penalty DDU, the plant operator employs a dynamic adjustment of the target weights by the anticipated benefits. We find that the value of $\Delta E^h$ is increasing as the algorithm progresses. This suggests that the constraints associated with the FR penalty are becoming less significant and that producers are directing their attention toward production rather than frequency regulation. This is due to the fact that the rate of decrease is dependent upon the degree of uncertainty present within the optimization problem. During the solution process of the iterative loop, \textbf{SP} continues to add new constraints to \textbf{MP}. At the beginning of the iteration, the problem uncertainty is considerable, and the model is unable to accurately estimate the percentage of FR objective terms at the optimal solution. However, with the addition of new constraints, the model uncertainty decreases, and thus the model can rapidly converge to the optimal objective function weights.

For the product structure DDU, there is a local optimum for the ratio of primary and secondary products. In the experiment we change the value of confidence $\gamma$, i.e., we change the values of $\zeta^h,\underline{\theta}^h,\bar{\theta}^h$. For a given $G$, a larger $\gamma$ brings a larger value of $\underline{\theta}^h,\bar{\theta}^h$, i.e., the upper and lower bounds of the by-product share are simultaneously higher. We solved for different $\gamma$ to obtain the optimal values of the model in Table \ref{tab2}. It can be found that the value of the objective function decreases and then increases as the ratio of main and by-products decreases. This is because appropriate output of by-products can effectively utilize redundant raw materials; however, by-products are less profitable, so excessive output will reduce profits.

\begin{table}[t]
\renewcommand{\arraystretch}{1.3}
\caption{Optimal values under different $\gamma$}
\centering
\begin{tabular}{ccc}
\toprule
\textbf{$\gamma$} & \textbf{By-products Tendency} & \textbf{Objective Function (CNY)}\\
\midrule
0.01 & $\downarrow \downarrow$ & -3319.52 \\
0.02 & $\downarrow$ & -3328.47 \\
0.05 & $\uparrow$ & \textbf{-3340.78} \\
0.10 & $\uparrow \uparrow$ & -3312.27 \\
\bottomrule
\end{tabular}
\label{tab2}
\end{table}

\begin{table*}[h]
\renewcommand{\arraystretch}{1.3}
\caption{MEANINGS OF PLACES AND TRANSITIONS IN THE MODEL}
\label{table_example}
\centering
\footnotesize
\begin{tabular}{p{4cm}|p{4cm}|p{4cm}|p{4cm}}
\hline
\multicolumn{2}{c|}{\textbf{Places}} & \multicolumn{2}{c}{\textbf{Transitions}} \\
\hline
P1: Buffer0 & P30: Crankshaft on CG & T1: Raw material input & T30: Production information input \\
P2: SM & P31: CG & T2: SM start & T31: CNC grinder start \\
P3: Raw material in SM & P32: CNC grinder & T3: SM running & T32: CG start \\
P4: Raw material in SM & P33: Crankshaft on CNC grinder & T4: SM start & T33: CG running \\
P5: SM & P34: Buffer3 & T5: SM running & T34: CNC grinder running \\
P6: Buffer1 & P35:  $\text{RAL}_2$ & T6: Millin, face milling start & T35: $\text{RAL}_2$ running \\
P7: Parts in milling, face milling & P36: Cylinder in $\text{RAL}_2$ & T7: Millin, face milling running & T36: $\text{RAL}_2$ start \\
P8: Millin, face milling & P37:  $\text{RAL}_2$ & T8: CNC grinder feed system start & T37: $\text{RAL}_2$ start \\
P9: Buffer2 & P38: Cylinder in $\text{RAL}_2$ & T9: CNC grinder start & T38: $\text{RAL}_2$ running \\
P10: Crankshaft on CG & P39: Buffer8 & T10: CG start & T39: Information processing \\
P11: Buffer4-1 & P40: Information control system & T11: CG start & T40: Sampling and transportation \\
P12: Main bearing in $\text{AAL}_1$ & P41: Information control system & T12: CG running & T41: Sampling and transportation \\
P13: Buffer5-1 & P42: Buffer9 & T13: $\text{AAL}_1$ start & T42: Sampling and transportation \\
P14: Semi-finished in machine & P43: $\text{AAL}_1$ & T14: $\text{AAL}_1$ running & T43: Sampling and transportation \\
P15: Buffer6-1 & P44: Main bearing in $\text{AAL}_1$ & T15: PAM start & T44: Cargo transportation \\
P16: Cylinder in $\text{AAL}_2$ & P45: $\text{AAM}_2$ & T16: PAM running & T45: $\text{AAM}_1$ start \\
P17: Crankshaft on CNC grinder & P46: Belt on $\text{AAM}_2$ & T17: $\text{AAL}_2$ start & T46: $\text{AAL}_1$ start \\
P18: Buffer4-2 & P47: Semi-finished in $\text{AAM}_2$ & T18: $\text{AAL}_2$ running & T47: $\text{AAL}_1$ running \\
P19: Crankshaft on CNC grinder & P48: Buffer11 & T19: CNC grinder & T48: $\text{AAL}_1$ running \\
P20: Engineer communication unit & P49: Buffer10 & T20: CNC grinder running & T49: BMM, $\text{AAM}_2$ start \\
P21: Buffer4-3 & P50: Main product in IDU & T21: Bearing output & T50: $\text{AAM}_1$ running \\
P22: Bearing in  $\text{RAL}_1$ & P51: BMM, $\text{AAM}_2$ & T22: $\text{RAL}_1$ start & T51: $\text{AAM}_1$ running \\
P23: Buffer5-2 & P52: AI detection unit & T23: $\text{RAL}_1$ running & T52: Sampling and transportation \\
P24: Semi-finished in machine & P53: Integrated test platform & T24: PAM start & T53: BMM, $\text{AAM}_2$ running \\
P25: Buffer6-2: Register & P54: Buffer7-2 & T25: PAM running & T54: BMM, $\text{AAM}_2$ running \\
P26: Cylinder in $\text{AAL}_2$ & P55: Worker inspection station & T26: $\text{AAL}_2$ start & T55: AI detection unit running \\
P27: Buffer7-1 & P56: Worker inspection station & T27: $\text{AAL}_2$ running & T56: AI detection unit start \\
P28: Defective in repair station & P57: Buffer12 & T28: Defective repair & T57: Main product output \\
P29: Engineer repair station & & T29: Defective transportation & T58: Main bearing input \\
\hline
\end{tabular}
\label{tab3}
\end{table*}

\section{Conclusion}
The implementation of a co-scheduling strategy for energy and production is crucial for discrete manufacturing. In this paper, a two-stage robust optimization model considering DDUs is proposed for energy management and equipment dispatch. The model needs to consider not only traditional DIU variables but also multiple types of DDU variables (ambiguous sets, probability distributions, etc.). After constructing and subsequently linearizing the model, we proposed a novel algorithm named DDCCG with fast convergence speed for solving this decision-dependent two-stage robust optimization problem. Ultimately, the resulting solutions are the optimal operating point for each device and the energy access scheme for the energy system. We conducted a case study by using a real-world engine assembly line modeled by Petri net method. The results show that our model effectively reduces the production cost by lowering the electricity demand without compromising the final product output. At the same time, our DDU model satisfies the frequency regulation objective better, achieving peak shaving as well. In addition, our approach contributed to the enhancement of resilience to risk in the production line. In conclusion, the findings of our research are of significant value in enhancing the efficiency of discrete industrial production and energy dispatch, with promising future applications.

\appendices

\section{Petri Net Node Description}
Engine assembly line Petri net nodes are described as TABLE \ref{tab3}.




%





\ifCLASSOPTIONcaptionsoff
  \newpage
\fi



%

\bibliographystyle{ieeetr}
\bibliography{reference}

%







\end{document}